\documentclass[12pt,reqno]{amsart}

\usepackage{graphics,graphicx,amsmath,amssymb,epsfig,stmaryrd,color,flafter,xspace,lscape,threeparttable,pdfpages}
\usepackage{csquotes,amsfonts,amsthm,geometry,array,booktabs,latexsym,caption,subfig,multirow,tabularx,rotating}
\usepackage[font=small,labelfont=bf]{caption}
\usepackage{changepage,enumitem}
\usepackage{amsmath}
\usepackage{amsfonts}
\usepackage{amssymb}
\usepackage{graphicx}
\usepackage{natbib}
\usepackage{dsfont} 
\usepackage{float}
\usepackage{bm}

\usepackage{tikz}
\usetikzlibrary{arrows}
\usetikzlibrary{decorations.markings}

\usepackage{algorithm}
\usepackage[noend]{algpseudocode}
\algblock{Input}{EndInput}
\algnotext{EndInput}
\algblock{Output}{EndOutput}
\algnotext{EndOutput}
\newcommand{\Desc}[2]{\State \makebox[2em][l]{#1}#2}
\makeatletter
\def\BState{\State\hskip-\ALG@thistlm}
\makeatother

\DeclareCaptionFont{tiny}{\tiny}
\everymath{\displaystyle}
\RequirePackage{setspace,indentfirst,epsfig,psfrag,ifthen,ifpdf}

\theoremstyle{plain}

\newtheorem{assumption}{Assumption}
\newtheorem{subassumption}{Assumption\normalfont}[assumption]

\newenvironment{assumption*}
{\ifnum\value{subassumption}=0 \stepcounter{assumption}\fi\subassumption}
{\endsubassumption}
\newenvironment{assumption+}[1]
{\renewcommand{\thesubassumption}{#1}\subassumption}
{\endsubassumption}
\theoremstyle{definition}
\newtheorem{assump}{Assumption}

\newtheorem{definition}{Definition}
\newtheorem*{definition*}{Definition}
\newtheorem{example}{Example}

\newtheorem{lemma}{Lemma}

\newtheorem{proposition}{Proposition}

\newenvironment{continued}[1][continued]{\begin{trivlist}
\item[\hskip \labelsep {\bfseries #1}]}{\end{trivlist}}

\numberwithin{equation}{section}
\newcommand{\norm}[1]{\left\lVert#1\right\rVert}
\newcommand{\Lmap}{\mathcal{L}}

\graphicspath{ {images/} }

\begin{document}
\title[Vector Lorenz Curves]{Lorenz map, inequality ordering and curves based on multidimensional rearrangements}
\author{Yanqin Fan, Marc Henry, Brendan Pass and Jorge A. Rivero}
\address{University of Washington, Penn State, University of Alberta and
University of Washington}
\thanks{The first version is of March 11, 2022. This version is of \today. The authors are grateful to Isaiah Andrews, Marco Scarsini, Xiaoxia Shi, John Weymark and three anonymous referees for helpful comments.
Corresponding author: Marc Henry: \texttt{marc.henry@psu.edu}
Department of Economics, The Pennsylvania State University,
University Park, PA 16802.}

\begin{abstract}
We propose a multivariate extension of the Lorenz curve based on multivariate rearrangements of optimal transport theory. We define a vector Lorenz map as the integral of the vector quantile map associated with a multivariate resource allocation. Each component of the Lorenz map is the cumulative share of each resource, as in the traditional univariate case. The pointwise ordering of such Lorenz maps defines a new multivariate majorization order, which is equivalent to preference by any social planner with inequality averse multivariate rank dependent social evaluation functional. We define a family of multi-attribute Gini index and complete ordering based on the Lorenz map. 
We propose the level sets of an Inverse Lorenz Function as a practical tool to visualize and compare inequality in two dimensions, and apply it to income-wealth inequality in the United States between 1989 and 2022.
\vskip20pt

\noindent\textit{Keywords}: Multidimensional inequality, Lorenz curve, Gini index, vector quantiles, optimal transport, majorization

\vskip10pt

\noindent\textit{JEL codes}: D63

\end{abstract}

\maketitle


\section*{Introduction}

The Lorenz curve, first proposed in \cite{lorenz:1905}, is a compelling visual and simple quantification tool for the analysis of dispersion in univariate distributions. It allows easy visualization of dispersion from the curvature of a convex curve and its distance from the diagonal. The diagonal itself is the Lorenz curve of a degenerate distribution-- an egalitarian allocation where all individuals have the same amount of resource. It also enables quick computations, reading off the curve, as it were, of the share of a resource held by the top or bottom of the allocation distribution for that resource. These features of the Lorenz curve account for much of its enduring appeal among practitioners, policy analysts and policy makers. This appeal is further enhanced by the relation between majorization and the pointwise ordering of Lorenz curves, which provides a way to visualize inequality comparisons between populations and within a given population between time periods. Comprehensive accounts are given in \cite{marshall:1979} and \cite{arnold:2018:book}.

The appealing properties of the Lorenz curve are well captured by the formulation given in \cite{gastwirth:1971}. In that formulation, the Lorenz curve is the graph of the Lorenz map, and the latter is the cumulative share of individuals below a given rank in the distribution, i.e., the normalized integral of the quantile function. The relation to majorization and the convex order follows immediately, as shown in section~C of \cite{marshall:1979}. As pointed out by \cite{arnold:110}, this makes the Lorenz ordering an uncontroversial partial inequality ordering of univariate distributions, and most open questions concern the higher dimensional case.

Dispersion in multivariate distributions is not adequately described by the Lorenz curve of each marginal, and a genuinely multidimensional approach is needed. Even for utilitarian welfare inequality, \cite{AB:1982} motivate the need for the multidimensional  approach initiated by \cite{fisher:1956}. More generally, the literature on multidimensional inequality of outcomes and its measurement is vast, as evidenced by many recent surveys, see for instance \cite{DL:2012}, \cite{AB:2014}, \cite{AZ:2020}. We only discuss it insofar as it relates to the Lorenz curve.

Multivariate extensions have been proposed for the Lorenz curve, most notably \citeauthor{Taguchi:72a} (\citeyear{Taguchi:72a}, \citeyear{Taguchi:72b}), \cite{arnold:1983}, and \citeauthor{koshevoy:1996} (\citeyear{koshevoy:1996}, \citeyear{koshevoy:1999})\footnote{More recently, subsequent to our work, \cite{HM:2022} also adopt a multivariate rearrangement approach to the definition of multi-attribute Lorenz curves. They adopt a center-outward approach (see \cite{HBCM:2021}), which is better suited to define notions of middle class.}. They are reviewed in \cite{marshall:1979} and \cite{SJ:2014} and discussed in more details in section~\ref{sec:comparison}, where we compare them to our proposal. We contribute to this literature with a vector version of the \cite{gastwirth:1971} formulation of the Lorenz curve. We provide an implementable criterion to measure and compare inequality in multivariate distributions, which emulates the features of the Lorenz curve that most contributed to its success.

The traditional \cite{gastwirth:1971} formulation of the Lorenz curve is an integrated quantile over the lowest ranked individuals. To simplify the argument in the univariate case, model the population as a continuum on~$[0,1]$ and suppose the distribution of incomes in the population is continuous. Then the \cite{gastwirth:1971} formulation can be thought of involving two stages. Take an income allocation~$Y$, which is a random variable on~$\mathbb R_+$ with cumulative distribution function~$F_Y$. First, reorder individuals in the population so that they are ranked in increasing incomes. Then compute the cumulative share of lowest ranked individuals by integrating~$F_Y^{-1}$ from~$0$ to~$r$ and dividing by the mean. The first step involves the probability integral transform~$F_Y(Y)$, which should be thought of in this context as a cardinal to ordinal transformation, since~$F_Y(Y)$ is uniform on~$[0,1]$, so cardinal information is purged, but~$F_Y$ is increasing, so ordinal information is preserved. Our proposal is based on a multivariate version of the cardinal to ordinal transformation involved in the first stage. The latter is the unique map that transforms a~$d$ dimensional allocation into a uniform one on~$[0,1]^d$, and is cyclically monotone\footnote{Existence and uniqueness are shown in \cite{mccann:1995}. See section~\ref{sec:lmaps} for details and definitions.}
and hence preserves ordinal information. This motivates our definition of the vector Lorenz map as the cumulative integral of the multivariate quantile of \cite{CGHH:2017}.

The vector Lorenz map we propose, therefore, is the vector of shares of each resource held by individuals below a given rank. The associated Lorenz inequality dominance criterion deems a multivariate allocation more equal if this share of resources is larger for each rank. Hence, our proposal shares the interpretation of the traditional Lorenz curve and Lorenz dominance. It also shares the desirable properties of the Lorenz curve and dominance ordering. Like the Lorenz zonoid of \citeauthor{koshevoy:1996} (\citeyear{koshevoy:1996}, \citeyear{koshevoy:1999}), it characterizes the distribution of an allocation (see section~\ref{sec:comparison} for a definition and discussion). Unlike the Lorenz zonoid, the vector Lorenz map we propose can be efficiently computed as an unconstrained convex optimization problem and connected to recent developments in computational optimal transport theory\footnote{An account of recent advances is given in \cite{peyre:2018}.}. Hence, the Lorenz dominance order we propose is an implementable inequality dominance criterion. Using recent advances on the asymptotic properties of multivariate quantiles, surveyed in \cite{Hallin:2022}, our Lorenz dominance criterion can be the basis for inequality dominance testing that accounts for sampling uncertainty. This contrasts our proposal with the growing literature on multivariate inequality dominance criteria proposed for finite populations. See for instance \cite{GM:2012}, \cite{banerjee:2016}, \cite{FG:2021} and references within.  

Other implementable inequality dominance criteria are proposed in the literature, in \cite{koshevoy:1995}, \cite{koshevoy:1996}, \cite{koshevoy:2007} and \cite{banerjee:2016} and other references surveyed in \cite{arnold:2018:book}. However, they do not provide an equivalence between the Lorenz dominance criterion and a class of compatible social evaluation functionals. An exception is \cite{GM:2012} and \cite{FG:2021} who give a comprehensive treatment of the special case of a finite population with a single cardinal transferable attribute combined with an ordinal non transferable one. We characterize the class of social evaluation functionals that are inequality averse in the sense that they are increasing in the Lorenz dominance order. We build on the multivariate extension of the \cite{Quiggin:92}-\cite{Yaari:87} rank dependent decision theory in \cite{GH:2012} to show that, as in \cite{weymark:1981} for the univariate case, social evaluation functionals are inequality averse if and only if they are rank dependent social evaluation functionals with attribute specific weights decreasing in ranks. We also characterize the class of transfers that increase inequality according to the Lorenz dominance criterion as rank preserving transfers of any attribute from a lower to a higher ranked individual. A special case of such transfers, which we call {\em monotone regressive transfers} weakly increase marginal inequality and dependence between attributes.

To visualize Lorenz dominance, we define an {\em Inverse Lorenz Function} at a given vector of resource shares as the fraction of the population that cumulatively holds those shares. It is characterized by the cumulative distribution function of the image of a uniform random vector by the Lorenz map. Hence, it is a cumulative distribution function by construction, like the univariate inverse Lorenz curve. In two dimensions, the~$\alpha$-level sets of this cumulative distribution function, which we call~$\alpha$-Lorenz curves, are non crossing downward sloping curves that shift to the south-west when inequality increases, as defined by the Lorenz ordering. For the cases, where allocations are not ranked in the Lorenz inequality dominance ordering, we propose a family of multivariate S-Gini coefficients based on our vector Lorenz map, with the flexibility to entertain different tastes for inequality in different dimensions. Finally, we propose an illustration to the analysis of income-wealth inequality in the United States between 1989 and 2022.


\subsubsection*{Plan of the paper}

In the first section, we define the Lorenz map, explain its computation, detail its properties and how it compares with alternative proposals. In section~\ref{sec:ineqorder}, we introduce the Lorenz dominance ordering, its characterization in terms of classes of social evaluation functionals and in terms of transfers compatible with it. Section~\ref{sec:appli} illustrates the implementation of our proposed tools, and the final section concludes.


\section{Vector Lorenz Map}

\subsection{Definition of the Lorenz map}
\label{sec:lmaps}

The Lorenz curve was originally proposed in \cite{lorenz:1905} to provide a graphical representation of inequality of distribution of a single resource. Let~$Y$ be a random variable on~$\mathbb R_+$ with cumulative distribution function~$F_Y$, which represents the allocation of a resource in a population. The population is modeled as the continuum~$[0,1]$. 

The Lorenz curve is traditionally defined as the set of points in~$[0,1]^2$, parameterized by~$y$, with coordinates
\begin{eqnarray}
\label{eq:trad}
    \left( F_Y(y),\frac{1}{\mu_Y}\int_0^yvdF_Y(v) \right),
\end{eqnarray}where~$\mu_Y$ is the expectation of~$Y$. See for instance page~149 of \cite{arnold:2018:book}. 
\cite{gastwirth:1971} points out that the Lorenz curve is given by the graph of the map on~$[0,1]$
\begin{eqnarray}
\label{eq:scalar}
\begin{array}{ccccc}
q &\mapsto & L_Y(q) & = & \frac{1}{\mu_Y} \int_0^qF_Y^{-1}(v) \, dv,
\end{array}
\end{eqnarray}
where~$F^{-1}(v):=\inf\{y: v\leq F_Y(y)\}$ is the traditional quantile function.
Formulation~(\ref{eq:scalar}) provides a closed form expression and simple interpretation: For each proportion~$q\in[0,1]$, the Lorenz map gives the cumulative share of the resource held by the poorest proportion~$q$ of the population. This relies on the well known fact that the quantile function~$F_Y^{-1}$ is the only increasing function such that for any uniformly distributed random variable~$V$ on~$[0,1]$, $F_Y^{-1}(V)$ is distributed identically to~$Y$. Hence, integrating~$F_Y^{-1}$ from~$0$ to~$q$ and normalizing produces the cumulative share held by the individuals ranked below~$q$.

Conversely, when~$F_Y$ admits a density, the probability integral transform~$V:=F_Y(Y)$ produces a uniformly distributed random variable~$V$ on~$[0,1]$, which preserves the ranks of individuals in the population. This holds because~$F_Y$ is an increasing map. Hence, the probability integral transform removes cardinal information (by producing a uniformly distributed outcome), while preserving ordinal information (by keeping the rank order of individuals in the population). The probability integral transform~$V:=F_Y(Y)$ is the rank associated with allocation~$Y$.

If~$F_Y$ is not continuous, then~$F_Y(Y)$ is no longer uniformly distributed (positive masses of individuals have identical ranks). However, it is still the case that for any uniformly distributed random variable~$V$ on~$[0,1]$, $F_Y^{-1}(V)=\inf\{y: V\leq F_Y(y)\}$ is distributed identically to~$Y$. Hence, the closed form solution for the Lorenz map~(\ref{eq:scalar}) still holds with the same interpretation: Integrating~$F_Y^{-1}$ from~$0$ to~$q$ and normalizing still produces the cumulative share held by the individuals ranked below~$q$.

Consider now an allocation~$X:=(X_1,\ldots,X_d)$ of~$d$ resources in the population. To analyze inequality in allocation~$X$, we can first look at inequality in each marginal allocation~$X_1,\ldots,X_d$, using the univariate Lorenz curves~$L_1:=L_{X_1},\ldots,L_d:=L_{X_d}$. However, this strategy disregards the effect of dependence. The latter is relevant to inequality, as can be trivially illustrated by the fact that for given wealth and income marginal allocations, the comonotonic allocation (the wealthier individuals have higher income) is more unequal than the admittedly unrealistic counter-monotonic allocation (the wealthier individuals have lower income).

To take dependence into account, we propose to emulate the 
\cite{gastwirth:1971} formulation by measuring cumulative shares of each resource for all individuals, {\em below a certain rank}. Conceptually, this is achieved in two steps. First, we find a transformation that removes cardinal information while preserving individual's ranking in the population, i.e., a cardinal to ordinal transformation. Then we integrate the shares of individuals with lowest rank. The difficulty here, of course, is the absence of a canonical order in~$\mathbb R^d$ to define the rank. 

As noted in \cite{FR:2017}, by Borel's Isomorphism Theorem\footnote{See for instance section~13.1 page~487 of \cite{Dudley:2002}.}, there exist measurable bijective maps~$T:[0,1]\rightarrow\mathbb R^d$ such that for any uniformly distributed random variable~$V$ on~$[0,1]$, $T(V)$ is distributed identically to~$X$. However, such maps are unsuitable cardinal to ordinal transformations for two main reasons. First, there is no known explicit construction, hence no way to compute them. Second, even if we could compute such a map, its choice would imply an implicit ad hoc aggregation of the different resources in allocation~$X$ in order to arrive at a scalar ranking of individuals in the population.

In order to avoid an implicit ad hoc aggregation of the different resources in~$X$, the cardinal to ordinal transformation must be between~$\mathbb R^d$ and~$[0,1]^d$. Hence, we model the population as a continuum on~$[0,1]^d$ and individual ranks are points in~$[0,1]^d$. The multivariate quantile transform, and its inverse (the cardinal to ordinal transform, or rank transform), must satisfy the same requirements as in the univariate case: It must map the uniform distribution (no cardinal information) to the distribution of the allocation, and it must be monotonic (so as to preserve ordinal information). The monotonicity of the quantile in the univariate case ensures that the cardinal to ordinal transformation does indeed preserve the rankings of individuals in the population. 

To construct an analogue of the \cite{gastwirth:1971} Lorenz curve formulation, we therefore need the cardinal to ordinal transformation to satisfy a form of multivariate monotonicity. The classical notion of monotonicity in~$\mathbb R^d$, also known as $2$-monotonicity of a map~$T:\mathbb R^d\rightarrow\mathbb R^d$, requires
\begin{eqnarray*}
    \left( T(x^\prime)-T(x) \right)^\top\left( x^\prime-x \right) & \geq & 0
\end{eqnarray*}
for any pair of vectors~$x,x^\prime\in\mathbb R^d$. It can be interpreted as monotonicity on average. For uniqueness of the cardinal to ordinal transformation, we need the stronger version of monotonicity, called {\em cyclical monotonicity}, which characterizes the gradients of convex functions and was introduced by \cite{Rockafellar:66}. Cyclical monotonicity requires
\begin{eqnarray*}
    \sum_{i=1}^K (T(x_i)-T(x_{i+1}))^\top\left( x_i-x_{i+1} \right) & \geq & 0
\end{eqnarray*}
for any~$K$, and any collection of vectors~$(x_1,\ldots,x_K)$, setting~$x_{K+1}=x_1$. Cyclical monotonicity also characterizes maps~$T$ that minimize distortion in the sense that in case the allocation~$X$ has continuous distribution with finite variance, $T$ minimizes~$\mathbb E\| X-T(X)\|^2$ among all the maps such that~$T(X)$ is uniformly distributed on~$[0,1]^d$.

The following definition summarizes the properties needed for a cardinal to ordinal transformation as a first step in the Lorenz map construction.
\begin{definition}[Vector quantile]
\label{def:VQ}
A vector quantile~$Q_X$ associated with random vector~$X$ on~$\mathbb R^d$ is a map~$Q_X:[0,1]^d\rightarrow\mathbb R^d$ with the following properties.
\begin{enumerate}
   \item For any uniformly distributed random variable~$U$ on~$[0,1]^d$, $Q_X(U)$ is distributed identically to~$X$.
    \item If the distribution of~$X$ is absolutely continuous, $Q_X$ is invertible and~$Q_X^{-1}(X)$ is uniformly distributed on~$[0,1]^d$.
    \item The map~$Q_X$ is cyclically monotone.
    \item When~$d=1$, $Q_X$ is the traditional quantile function (This is automatically satisfied when~(1) and~(3) hold).
    
\end{enumerate}
\end{definition}
As shown in \cite{mccann:1995}, there exists a transformation that conforms with definition~\ref{def:VQ} and it is unique in the sense that two such transformations are equal almost everywhere. It is proposed as a vector quantile notion in \cite{CGHH:2017}, and we will refer to it as the vector quantile associated with~$X$. 

Once we model the population as a continuum on~$[0,1]^d$, interpret each point on~$[0,1]^d$ as a rank, and define the vector quantile~$Q_X$ as a multidimensional rearrangement of the allocation~$X$ in rank order, we simply integrate the quantile over the lowest ranks to define a multivariate version of the \cite{gastwirth:1971} formulation of the Lorenz curve.

\begin{definition}[Lorenz map]
\label{def:Lorenz}
Let~$U$ be a uniformly distributed random vector on~$[0,1]^d$, and let~$X:=(X_1,\ldots,X_d)$ be an allocation, i.e., a random vector on~$\mathbb R_+^d$ with finite mean~$\mu=(\mu_1,\ldots,\mu_d)$. Call~$\tilde X$ the normalized version of~$X$, i.e.,
\begin{eqnarray*}
    \tilde X:=\left( \frac{X_1}{\mu_1},\ldots,\frac{X_d}{\mu_d}\right),
\end{eqnarray*}
and let~$Q_{\tilde X}$ be the vector quantile of~$\tilde X$.
The \emph{Lorenz map} of allocation $X$ is the
vector-valued function $\Lmap_X:[0,1]^d \to [0,1]^d$ defined for each~$r:=(r_1,\ldots,r_d)\in[0,1]^d$ by 
\begin{eqnarray}
\label{eq:Lorenz}
\Lmap_X(r_1,\ldots,r_d) & = & \int_0^{r_1}\!\!\!\cdots\!\int_0^{r_d} Q_{\tilde X}(u_1,\ldots,u_d)du_1\ldots du_d.
\end{eqnarray}
\end{definition}

The transformation of~$X$ into its normalized version~$\tilde X$ prior to integrating the vector quantile is required to remove dependence of the Lorenz map of definition~\ref{def:Lorenz} on units of measurements. Different resources, such as earnings and health, may not be measured with the same units of measurement. The transformation into~$\tilde X$ makes the allocation unit free. Hence the Lorenz map satisfies ratio-scale invariance (i.e., invariance to rescaling of the different attributes, or change of units of measurement). Section~\ref{sec:norm} discusses an alternative unnormalized version of the definition in the spirit of \cite{Shorrocks:83}. 

When~$X$ has absolutely continuous distribution~$P_X$, its quantile function is~$P_X$-almost everywhere invertible (see for instance theorem~2.1 in \cite{CGHH:2017}). In that case, the transformation~$U=Q_X^{-1}(X)$ is the vector analogue of the probability integral transform~$V=F_Y(Y)$ discussed above. The random vector~$U=Q_X^{-1}(X)$ is uniformly distributed on~$[0,1]^d$, and is the vector rank of the individual with endowment~$X$, in the terminology of \cite{CGHH:2017}. The Lorenz map of definition~\ref{def:Lorenz} can then be rewritten as:
\begin{eqnarray}
\label{eq:KM}
\Lmap_X(r) & = & \mathbb{E}\left[ \tilde X\;\mathds{1}\left\{Q_{\tilde X}^{-1}(\tilde X)\leq r\right\}\right].
\end{eqnarray}
This clarifies the interpretation of~$\Lmap_X(r)$ as the cumulative share of all individuals with vector rank below~$r$ in the partial order of~$\mathbb R^d$. 

In the scalar case discussed above, inverting the Lorenz curve~$L_Y$ defined in~(\ref{eq:scalar})
yields the inverse Lorenz curve 
\begin{eqnarray}
\label{eq:scalar-ilf}
\begin{array}{ccccccc}
L_Y^{-1}(y) & = & \int_0^{L_Y^{-1}(y)} dv & = & \int_0^1 \, \mathds{1}\{L_Y(v) \leq y\} \, dv & = & \mathbb P(L_Y(V)\leq y),
\end{array}
\end{eqnarray}
where the probability is taken with respect to a uniformly distributed random variable~$V$ on~$[0,1]$.
The scalar inverse Lorenz curve at~$y$ is therefore shown in~(\ref{eq:scalar-ilf}) to be equal to the maximum proportion of the population with cumulative share of the resource equal to~$y$. In the vector case, the analogue of the right-hand side of~(\ref{eq:scalar-ilf}) can still be used to define an Inverse Lorenz Function.
\begin{definition}[Inverse Lorenz Function]
\label{def:ILS} The Inverse Lorenz Function (ILF) of a random
vector~$X$ is the function $l_X:[0,1]^d \to [0,1]$ defined for each~$z=(z_1,\ldots,z_d)\in[0,1]^d$ by~$l_X(z):=\mathbb P(\Lmap_X(U)\leq z),$
where~$z=(z_1,\ldots,z_d)\in[0,1]^d$, inequality~$\leq$ is understood component-wise, and the probability is taken with respect to the uniform random vector~$U$ on~$[0,1]^d$.
\end{definition}
The expression above is no longer the mathematical inverse of the Lorenz map~$\Lmap_X$, but it can still be interpreted as the share of the population with cumulative shares of all resources equal to a predetermined proportion~$z=(z_1,\ldots,z_d)$.

\subsection{Computation and examples}
\label{sec:comp}

\subsubsection{Computation}

We now give a step-by-step method to compute the Lorenz map of a discrete distribution, which may be the allocation in a finite population, or the empirical distribution of a (possibly weighted) sample from an underlying (possibly mixed discrete-continuous) distribution. The full algorithm and a step-by-step guide to implementation in \emph{R} are given in appendix~\ref{sec:algo}. 

Let~$X$ be a random vector in~$\mathbb R_+^d$ with discrete distribution. The probability mass function of the distribution of~$X$ is given by~$\{(x^1,w_1),\ldots,(x^n,w_n)\}$, where~$x^1,\ldots,x^n$ are vectors in~$\mathbb R_+^d$ and~$w_1,\ldots,w_n$ are positive scalar weights summing to~$1$.

\begin{enumerate}

    \item First, normalize the allocation vector~$X$ and form~$\tilde X:=(X_1/\mu_1,\ldots,X_d/\mu_d)$, where~$X_j$ is the~$j$-th coordinate of~$X$ and~$\mu_j=\Sigma_{i=1}^nw_ix_j^i$ is the mean of~$X_j$, for each~$j=1,\ldots,d$. The issue of normalization is discussed in section~\ref{sec:norm}.

    \item Then compute the vector quantile~$Q_{\tilde X}$ of~$\tilde X$. According to definition~\ref{def:VQ} (requirements~(2) and~(4)), the vector quantile~$Q_{\tilde X}$ must satisfy the following requirements:
    \begin{enumerate}
        \item For any uniformly distributed random variable~$U$ on~$[0,1]^d$, $Q_{\tilde X}(U)$ is distributed identically to~$\tilde X$. Hence:
        \begin{enumerate}
            \item For all~$u\in[0,1]^d$, $Q_{\tilde X}(u)\in\{x^1,\ldots,x^n\}$;
            \item For all~$i=1,\ldots,n,$ 
        \begin{eqnarray}\label{eq:vqcon2}
        W_i & : = & Q_{\tilde X}^{-1}(x^i)=\{u\in[0,1]^d: Q_{\tilde X}(u)=x^i \}
        \end{eqnarray}
        has measure~$w_i$.
        \end{enumerate}
        
        \item The map~$Q_{\tilde X}$ is cyclically monotone. Hence, by \cite{Rockafellar:66}, there is a convex function~$\psi_{\tilde X}:[0,1]^d\rightarrow\mathbb R$ such that~$Q_{\tilde X}$ is almost everywhere equal to the gradient of~$\psi_{\tilde X}$. 
    \end{enumerate}
    Since~$Q_{\tilde X}$ takes a finite number of values and is constant and equal to~$x^i$ on each~$W_i$, the computation of~$Q_{\tilde X}$ is equivalent to the computation of the partition of~$[0,1]^d$ in regions~$W_1,\ldots,W_n$. As~$Q_{\tilde X}$ is the gradient of the convex function~$\psi_{\tilde X}$ and is constant on each of the~$W_i$, $\psi_{\tilde X}$ is affine on each of the~$W_i$, and each~$W_i$ is a convex polytope in~$[0,1]^d$. \cite{AHA:1998} show
    \begin{eqnarray*}
        \psi_{\tilde X}(u) & = & \max_{i=1,\ldots,n}\{u^\top x^i - h^i \},
    \end{eqnarray*}
    where~$h:=(h^1,\ldots,h^n)$ solves the convex optimization program
    \begin{eqnarray}
    \label{eq:Newton}
        \min \; \left\{\sum_{i=1}^n  w_ih^i + \int_{[0,1]^d} \max_{k=1,\ldots,n}\{u^\top x^k - h^k \} du \right\}.
    \end{eqnarray}
     

    The algorithm minimizes~(\ref{eq:Newton}) to find~$h$, from which the 
    regions~$W_i$ are obtained as
    \begin{eqnarray}\label{eq:cxcells}
        W_i^h & = & \{u\in[0,1]^d: u^\top x^i-h^i\geq u^\top x^j-h^j,1\leq j\leq n\}.
    \end{eqnarray}
    The first order conditions of~\eqref{eq:Newton} fulfill requirement~(2)(a)(ii), i.e., for all $i = 1,\dots,n$,   
    \begin{eqnarray}\label{eq:constraint}
         \lambda(W_i^h) & = & w_i,
    \end{eqnarray}
 where~$\lambda$ is the Lebesgue measure. Finally, the vector quantile~$Q_{\tilde X}$ is the piecewise constant map that takes values~$x^i$ on each~$W_i:=W_i^h$, ~$i=1,\ldots,n$.
    
\item Once we have computed the vector quantile map~$Q_{\tilde X}$, the Lorenz map at~$r:=(r_1,\ldots,r_n)\in[0,1]^d$ is obtained straightforwardly as the integral of the piece-wise constant map~$Q_{\tilde X}$ over~$[0,r]:=[0,r_1]\times\ldots\times[0,r_d]$:  
\begin{eqnarray}\label{eq:EmpLorenzMap}
\Lmap_X(r) & = & \sum_{i=1}^n \lambda\left( W_i\cap [0,r]\right)\;x^i,
\end{eqnarray}
where the term in $\lambda$ is the ordinary area of the convex polytope formed by the intersection of the cell $W_i$ and the rectangle $[0,r]$.

\item Finally, $\Lmap_X$ can be used to generate a pseudo sample~$\{\Lmap_X(U_1),\ldots,\Lmap_X(U_m)\}$, where~$\{U_1,\dots,U_m\}$ is a uniformly distributed random sample or any pseudo-random (a.k.a. minimum discrepancy) sequence that approximates the uniform distribution on~$[0,1]^d$. The Inverse Lorenz Function~$l_X$ can then be approximated with the empirical distribution of this pseudo-sample:
\begin{eqnarray}\label{eq:EmpILF}
l_m(z) & := & \frac{1}{m} \sum_{j = 1}^m \mathds{1}\{\Lmap_X(U_j) \leq z\},
\>\>\> z \in [0,1]^d.
\end{eqnarray}
\end{enumerate}

\subsubsection{Examples}

To illustrate the definition and the computation of the Lorenz map, we now explore examples of specific allocations and compute the corresponding Lorenz maps. First, we illustrate the computation of the Lorenz map for a discrete allocation. 

\begin{example}[Discrete allocations]
\label{ex:discrete}
Let~$X$ be the allocation with probability mass function~$\{(x^1,1/n),\ldots,(x^n,1/n)\}$. We select the support points~$(x^1,\ldots,x^n)$ as the realizations of~$n$ i.i.d. draws from the bivariate standard normal distribution. The vector quantile~$Q_{\tilde X}$ of the normalized allocation~$\tilde X$ is characterized by its value~$x^i$ on the convex polygon~$W_i$, such that~$(W_1,\ldots,W_n)$ form the partition of~$[0,1]^2$ shown on the left panel of figure~\ref{fig:discrete}. As shown on the right panel of figure~\ref{fig:discrete}, the Lorenz map at~$r=(r_1,r_2)\in[0,1]^2$ is equal to the sum of the~$x^i$'s times the area of~$W_i$ intersected with~$[0,r_1]\times[0,r_2]$.
\end{example}

Next, we consider the special case, where all individuals are endowed with the same quantity of resources. 

\begin{example}[Identical allocations]
\label{ex:egal}
Let~$X$ be the constant allocation~$X=(1,1,...,1)$. Then~$Q_X (u)=(1 , 1,...,1)$ for all~$u\in \lbrack
0,1]^{d}$, so that~$\Lmap_X(u)$ is a~$d$-vector with identical entries~$u_{1}u_{2}\cdots u_d$. 
The image of~$\Lmap_X$ is the diagonal in~$[0,1]^d$. 
The Inverse Lorenz Function~$l_X(z)$ of~$X$ is~$0$ when~$z_1z_2\cdots z_d=0$. For~$d\geq 1$ and~$(z_{1},z_{2},...z_d)\in (0,1]^{d}$, and letting~$\underline z:=\min \{z_{1},z_{2},...,z_d\}$, the Inverse Lorenz Function~$l_X(z)$ of~$X$ is  
\begin{eqnarray*}
l_X(z) & = & \mathbb{P}(U_{1}U_{2}\cdots U_d\leq z_{1},U_1U_2\cdots U_d\leq z_{2},...,U_1U_2\cdots U_d\leq z_{d}) \\
& = & \mathbb{P}(U_1U_2\cdots U_d\leq \underline z) \\
& = & \underline z\sum_{k=1}^d \frac{(-1)^{k-1}}{(k-1)!}[\log(\underline z)]^{k-1}.
\end{eqnarray*}
\end{example}

\begin{figure}[H]
\centering
\includegraphics[width = 0.5\textwidth]{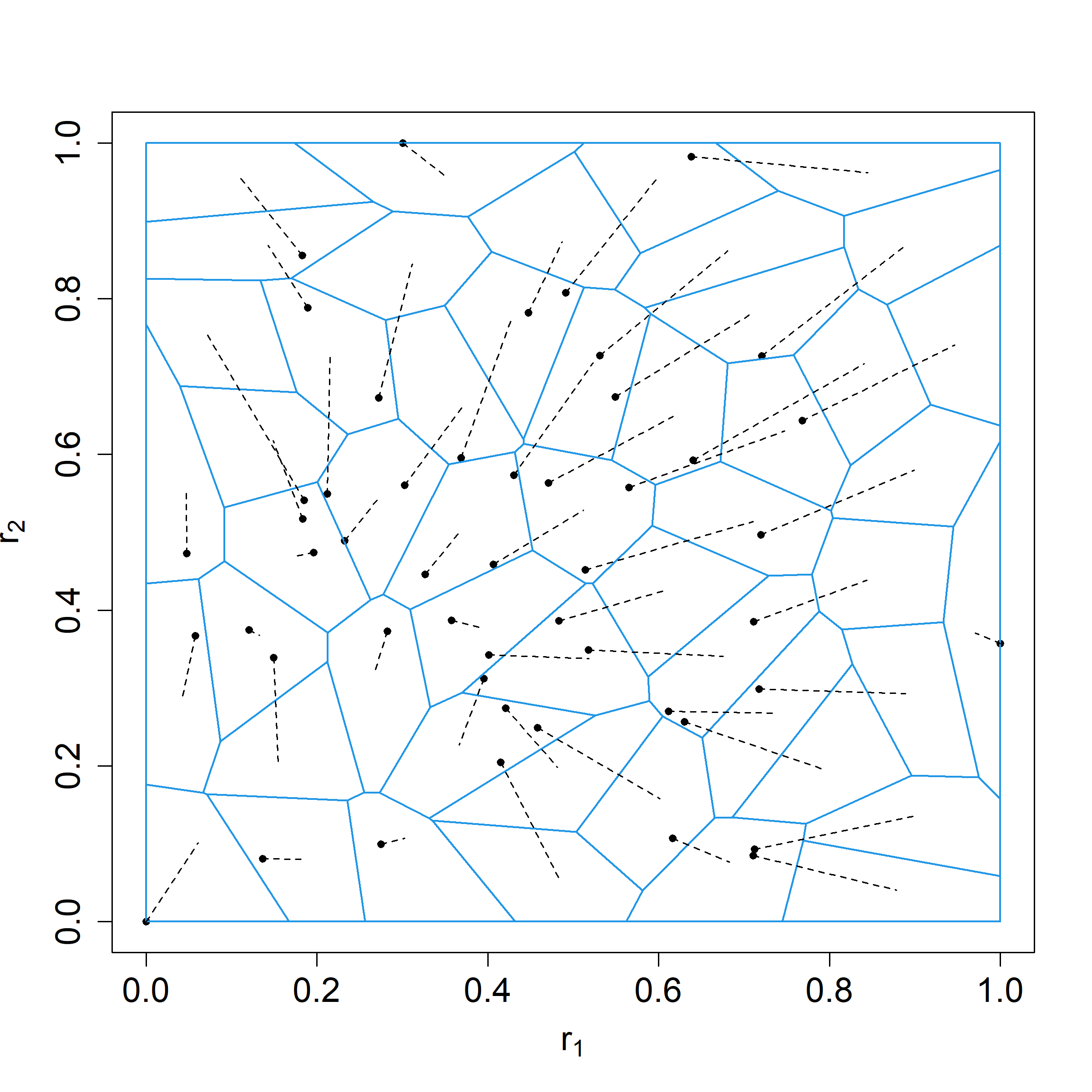}\includegraphics[width = 0.5\textwidth]{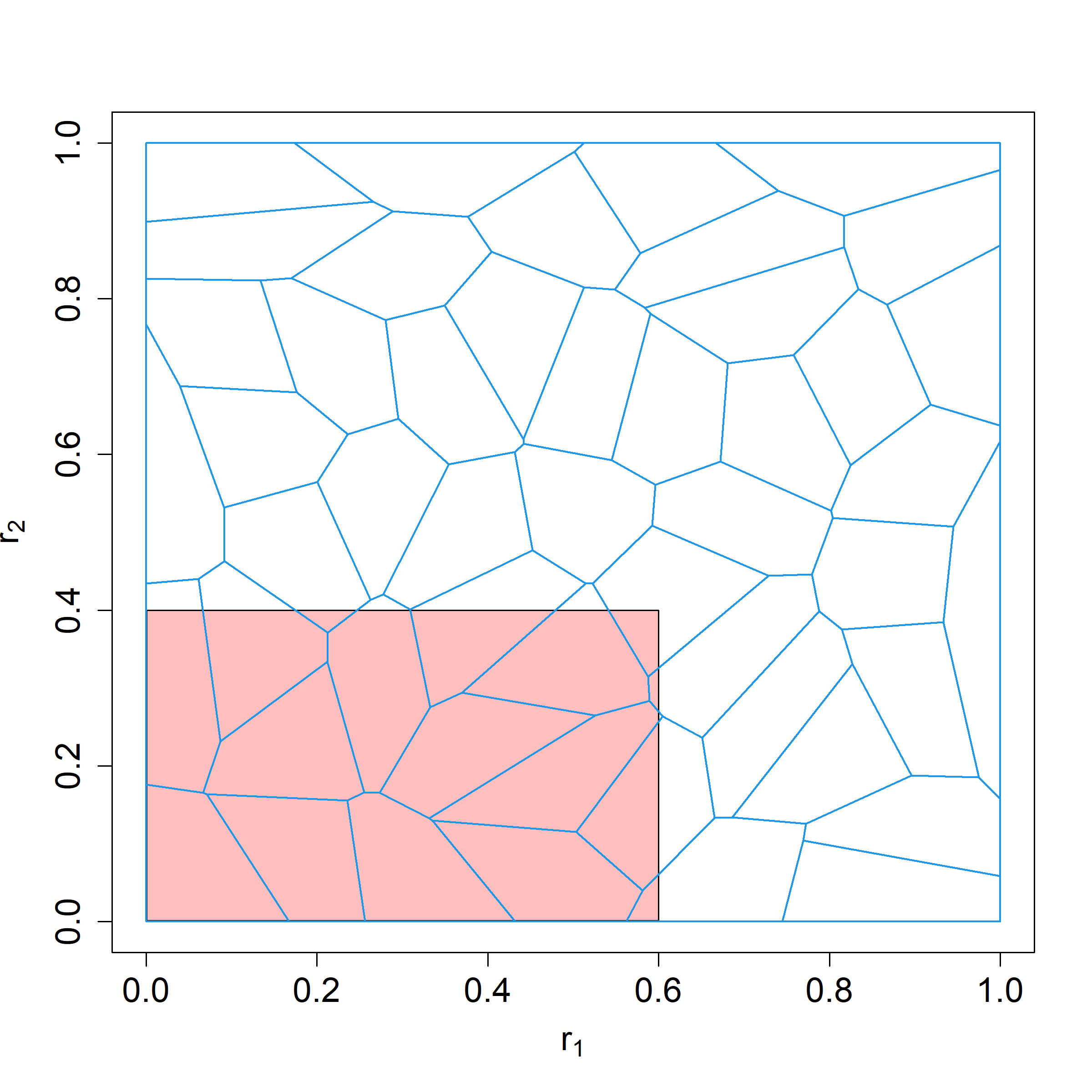} 
\caption{The left panel is a visualization of the quantile map. On each element of the blue partition of~$[0,1]^2$, the quantile map takes the value of the sample point. The sample point is indicated by a black dot and it is connected to the corresponding region by a dashed line. Note that sample points are drawn from the standard bivariate normal, then shifted and rescaled to fit into~$[0,1]^2$. The right panel shows a visualization of the computation of the Lorenz map at~$r:=(0.6,0.4)$. The latter is equal to the sum of areas of the intersection between the red shaded area and the elements of the blue partition weighted by the corresponding sample realization.}
\label{fig:discrete}
\end{figure}

We also check that our definition is compatible with scalar definitions when all resources are independently distributed.

\begin{example}[Independent Resources]
\label{ex:ind}
Let the components~$X_1,...,X_d$ of $X$ be independent with marginal
Lorenz curves~$L_{1},...,L_d$, respectively. Then, the $i$th component of the Lorenz map is $L_i(r_i) \Pi_{j=1,j\neq i}^d r_j$. %
This expression of the Lorenz map has the following interpretation. Consider the first component~$r_2\cdots r_dL_1(r_1)$. The share of resource~$1$ held by people with multivariate rank in~$[0,r_1]\times [0,1]^{d-1}$ is the marginal share, equal to the marginal Lorenz curve. Since the resources are independent, this share is uniformly distributed along the other dimensions, so that people with ranks in~$[0,r_1]\times[0,r_2] \times \cdots \times [0,r_d]$ command a share $r_2\cdots r_dL_1(r_1)$. The other components are interpreted analogously. When $d=2$ and ~$r_1=
1 $, the Lorenz map takes values~$\Lmap_X(1,r_2) = (r_2,L_2(r_2))$. That is,
the image of~$\{(1,r_2): 0 \leq r_2 \leq 1\}$ under~$\Lmap_X$ is the marginal Lorenz curve~$L_2$
of the second resource~$X_2$ (and symmetrically when~$r_2=1$).

The Inverse Lorenz Function~$l_X(z)$ of allocation~$X$ with independent components is 
\begin{eqnarray*}
l_X(z) = \mathbb{P}(\, U_{2}\cdots U_dL_1(U_{1})\leq z_1, U_{1}U_3\cdots U_dL_2(U_{2})\leq z_2,....,U_1\cdots U_{d-1} L_d(U_d)\leq z_d) \\
= \int_{[0,1]^{d-1}}\min\left\{ \frac{z_1}{L_1(u_1)\Pi_{k\neq 1,d}u_k}, \frac{z_2}{L_2(u_2)\Pi_{k\neq 2,d}u_k} ,...,l_d\left(\frac{z_d}{\Pi_{k\neq d} u_k}\right) \right\}du_1...du_{d-1},
\end{eqnarray*}
where~$l_d$ is the univariate inverse Lorenz curve of~$X_d$.
\end{example}

Next, we derive the Lorenz map in the case of allocations~$X=(X_1,X_2)$ with the same components, i.e.,~$X_1=X_2$ almost surely. We stick to~$d=2$ for notational simplicity.

\begin{example}[Comonotonic Resources]
\label{ex:comon}
Consider bivariate comonotonic allocations. Let the components~$X_1$ and~$X_2$ of the allocation~$X$ be almost surely equal. Then,~$X_1$ and~$X_2$ have identical distributions. 
Since the distribution of $X=(X_1,X_2)$ concentrates on the line $x_1=x_2$, the vector quantile depends on $u$ only through $u_1+u_2$ (it is an index cost in the terminology of \cite{CMP:2017}). More precisely, it is $(u_1,u_2)\mapsto (\psi'(u_1+u_2),\psi'(u_1+u_2))$, where~$z \mapsto \psi'(z)$ is the optimal transport map from~$\sigma$ to the distribution of~$X_1$, where~$\sigma
$ has density on $[0,2]$ given by $1-|1-z|$.
Each component of the Lorenz curve is then given by
\begin{eqnarray*}
\Lmap_1(r_1,r_2) = \Lmap_2(r_1,r_2) &=&\int_0^{r_2}\int_0^{r_1}\psi'(u_1+u_2)du_1du_2 =\int_0^{r_2}[\psi(u_2+r_1) -\psi(u_2)]du_2.
\end{eqnarray*}
In case~$X_1$ and~$X_2$ are uniformly distributed on~$[0,2]$, the optimal transport map~$\psi'$ for~$z <1 $ is given by $\psi'(z) = z^2$, so that $\psi(z) =z^3/3$.  We then have,
	\begin{equation*}
	\Lmap_1(r_1,r_2) = \frac{r_1^3r_2}{3}+\frac{r_1r_2^3}{3}+\frac{r_1^2r_2^2}{2},
	\end{equation*}
when~$r_1+r_2\leq 1$, and 
\begin{eqnarray*}
\Lmap_1(r_1,r_2) = \, \frac{2}{3}(r_1+r_2)^3-\frac{1}{12}(r_1+r_2)^4-(r_1+r_2)^2-\frac{r_1^4}{12}-\frac{r_2^4}{12}+\frac{2}{3}(r_1+r_2)-\frac{1}{6},
\end{eqnarray*}
when~$r_1+r_2>1$. The image of this Lorenz map is once again the diagonal in $[0,1]^2$.

Letting~$R~U[0,1]^2$, the Inverse Lorenz Function~$l_X(z)$ of allocation~$X=(X_1,X_2)$ with~$X_2=X_1$ almost surely, is
\begin{eqnarray*}
l_X\left( z\right) &=&\mathbb{P} \left( \Lmap_{1}\left( R\right) \leq z_{1},\Lmap_{2}\left(
R\right) \leq z_{2}\right) \\
&=&\mathbb{P} \left( \Lmap_{1}\left( R\right) \leq \min \left\{ z_{1},z_{2}\right\}
\right) \\
&=& h\left( \min \left\{ z_{1},z_{2}\right\} \right) ,
\end{eqnarray*}
where $\Lmap_j(R)$ is the~$j$-th component,~$j=1,2$, of~$\Lmap_X(R)$, and~$h$ is the distribution function of~$\Lmap_{1}\left( R\right) $. 
\end{example}

\subsection{Relative scale and normalization}
\label{sec:norm}

Let~$X$ be the original allocation. Normalizing~$X$ into~$\tilde X$ as in definition~\ref{def:Lorenz} by dividing each component by its mean, removes any sensitivity to (changes in) units of measurements. It is a standard approach to achieve ratio-scale invariance. See for instance \citeauthor{banerjee:2010} (\citeyear{banerjee:2010}, \citeyear{banerjee:2016}). However, by construction, it comes with the disadvantage of removing scale effects. In the univariate case, \cite{Shorrocks:83} proposes to eschew normalization in order to take scale effects into account in the measurement of inequality. 

When units of measurement are not a concern, an alternative definition without the feature described above is defined as
\begin{eqnarray*}
    \int_0^{r_1}\!\!\!\cdots\!\int_0^{r_d} Q_X(u_1,\ldots,u_d)du_1\ldots du_d,
\end{eqnarray*}
where~$Q_X$ is the vector quantile of the original allocation~$X$ (not the normalized one). 
This alternative version also allows the weighting of different resources according to a priori importance to overall inequality. 
Call~$X^\epsilon=(\lambda_1(\epsilon)X_1,\ldots,\lambda_d(\epsilon)X_d)$ the suitably rescaled version of the initial allocation~$X$. Define the sequence of weights~$(\lambda_1(\epsilon),\ldots,\lambda_d(\epsilon))$ in such a way that~$\lambda_{k+1}(\epsilon)/\lambda_k(\epsilon)\rightarrow0$ as~$\epsilon\rightarrow0$. In this way, resources are ordered in decreasing importance to inequality, and we can entertain the extreme lexicographic case, where~$\epsilon\rightarrow0$.

It follows from \cite{CGS:2010} that, when~$X$ has an absolutely continuous distribution, as~$\epsilon$ tends to~$0$, the alternative Lorenz map tends to the map
\begin{eqnarray*}
    \int_0^{r_1}\!\!\!\cdots\!\int_0^{r_d} Q^{KR}_X(u_1,\ldots,u_d)du_1\ldots du_d,
\end{eqnarray*}
where~$Q_X^{KR}$ is the inverse of the Knothe-Rosenblatt transform~$T^{KR}:\mathbb R^d\mapsto[0,1]^d$ of the original allocation~$X$ proposed by \cite{Rosenblatt:52} and \cite{Knothe:57}, and defined as follows
\begin{eqnarray*}
    T^{KR}(x_1,x_2,\ldots,x_d) & := & 
    \left[ 
    \begin{array}{c}
    F_{X_1}(x_1) \\
    F_{X_2\vert X_1}(x_2\vert x_1) \\
    \vdots \\
    F_{X_d\vert X_1,\ldots,X_{d-1}}(x_d\vert x_1,\ldots,x_{d-1})
    \end{array}
    \right].
\end{eqnarray*}
The Knothe-Rosenblatt quantile map is the only multivariate quantile map from the uniform on~$[0,1]^d$ to~$\mathbb R^d$ proposed in the literature other than optimal transport based vector quantiles as in definition~\ref{def:VQ}. The result above shows that the Knothe-Rosenblatt quantile is not a good alternative to vector quantiles of definition~\ref{def:VQ} to base an integrated quantile definition for the Lorenz map, since it relies on an a priori lexicographic ordering of the different resources in the allocation.

\subsection{Properties and comparisons with other multivariate Lorenz concepts}

In this section, we detail previous proposals for multivariate extensions of the Lorenz curve and list the properties that distinguish our proposal from the former. 

\subsubsection{Alternative multivariate Lorenz proposals}
\label{sec:comparison}

Until now, the development of multivariate extensions of the Lorenz curve was hampered by the lack of simple multivariate analogues of ranks and quantiles. Early proposals for bivariate extensions of the Lorenz curve in \citeauthor{Taguchi:72a} (\citeyear{Taguchi:72a},\citeyear{Taguchi:72b}) and \citeauthor{arnold:1983} (\citeyear{arnold:1983},\citeyear{arnold:1987}) are based on a direct ad-hoc extension of the traditional formula given in~(\ref{eq:trad}).
Let~$(x_1,x_2)\mapsto F(x_1,x_2)$ be the CDF of a bivariate allocation with density~$f$ and mean~$(\mu_1,\mu_2)$. \citeauthor{Taguchi:72a} (\citeyear{Taguchi:72a},\citeyear{Taguchi:72b}) proposes the bivariate Lorenz surface $L:[0,1]^2\rightarrow[0,1]$ defined implicitly by~$(s,t,L(s,t)):=$
\begin{eqnarray}
\label{eq:Taguchi}
    \left ( \, F(x_1,x_2) \, , \, \frac{1}{\mu_1}\int_0^{x_1}\!\!\!\int_0^{x_2}u_1f(u)du \, , \, \frac{1}{\mu_2}\int_0^{x_1}\!\!\!\int_0^{x_2}u_2f(u)du \, \right).
\end{eqnarray}
In order to treat both dimensions of the allocation symmetrically, \citeauthor{arnold:1983} (\citeyear{arnold:1983},\citeyear{arnold:1987}) proposes the alternative Lorenz surface parameterized by~$(x_1,x_2)$ as the set of points
\begin{eqnarray}
\label{eq:Arnold}
    \left( \, F_1(x_1) \, , \, F_2(x_2) \, , \, \frac{1}{\mu_{12}}\int_0^{x_1}\!\!\!\int_0^{x_2}u_1u_2f(u)du \, \right),
\end{eqnarray}
where~$F_1$ and~$F_2$ are the marginal CDFs associated with~$F$, and~$\mu_{12}$ is the expectation of the product~$X_1X_2$.
A closed form solution, given in \cite{SJ:2014}, makes the Lorenz surface~(\ref{eq:Arnold}) amenable to parameterization and statistical analysis. However, it does not share the interpretation or any of the properties of the univariate Lorenz curve. 

A more successful proposal in that respect, is the Lorenz zonoid of \cite{koshevoy:1996}. Again, take~(\ref{eq:trad}) in the univariate case as the point of departure. It associates a fraction~$p$ of the population to the share of the resource collectively held by the poorest fraction~$p$ of the population. \cite{koshevoy:1996} eschew the need to order the population by associating with a fraction~$p$ of the population the share of resources held by any group of individuals making up a fraction~$p$ of the population, poor, rich, or mixed. The lower bound is the share held by the poorest individuals (the traditional Lorenz curve), and the upper bound is the share held by the richest individuals (a reverse Lorenz curve). The Lorenz zonoid is defined in \cite{koshevoy:1996} as the collection of all such shares for each fraction of the population. It is a convex region in~$[0,1]^2$ bounded below by the Lorenz curve and above by the reverse Lorenz curve. More precisely, the Lorenz zonoid is defined as the set of points
\begin{eqnarray*}
    L(Y) & := & \left\{ \left( \, \int_0^\infty \phi(v)dF_Y(v) \, , \, \frac{1}{\mu_Y}\int_0^\infty v\phi(v)dF_Y(v) \, \right) \, : \; \phi\in\Phi \right\},
\end{eqnarray*}
where the function~$\phi$ ranges over the set~$\Phi$ of measurable functions from~$\mathbb R_+$ to~$[0,1]$. The lower (resp. upper) bound is obtained with the collection of functions~$\phi(v):=\mathds 1\{v\leq y\}$ (resp. $\mathds 1\{v> y\}$), $y\in\mathbb R_+$.

Since the definition of the Lorenz zonoid does not rely on ranks or quantiles, the extension to higher dimensions is straightforward. Let~$\Phi$ now be the set of measurable functions from~$\mathbb R_+^d$ to~$[0,1]$ and~$X=(X_1,\ldots,X_d)$ be a multivariate allocation with CDF~$F$ and mean~$(\mu_1,\ldots,\mu_d)$. The Lorenz zonoid of \cite{koshevoy:1996} is defined as the set of points~$L(X):=$
\begin{eqnarray*}
    \label{eq:zonoid}
    \left\{ \left( \, \int_0^\infty \phi(u)dF(u) \, , \, \frac{1}{\mu_1}\int_0^\infty u_1\phi(u)dF(u) \, , \, \cdots \, , \, \frac{1}{\mu_d}\int_0^\infty u_d\phi(u)dF(u) \, \right) \, : \; \phi\in\Phi \right\}.
\end{eqnarray*}
The Lorenz zonoid is an American football-shaped region in~$[0,1]^{d+1}$ with poles at points~$(0,\ldots,0)$ and~$(1,\ldots,1)$. The Lorenz surface of \citeauthor{Taguchi:72a} (\citeyear{Taguchi:72a},\citeyear{Taguchi:72b}) is a subset of the Lorenz zonoid, obtained when~$\Phi$ is restricted to the set of functions~$\phi_x(\cdot):=\mathds{1}\{\cdot\leq x\}$, all~$x\in\mathbb R^2$. A function~$\phi\in\Phi$ defines a point in the zonoid. The interpretation is simple in the case of indicator functions. The latter pick out specific groups of individuals in the population and the corresponding point in the zonoid has first coordinate equal to the fraction of the population involved. The other coordinates are the shares of each of the resources held by this group of individuals.

Definition~4.1 of \cite{banerjee:2016} proposes a multivariate inequality ordering in the finite population case. The analogue multivariate Lorenz map in the general case of a (possibly mixed discrete continuous) vector allocation~$X$ with normalized version~$\tilde X$ can be defined\footnote{We thank Xiaoxia Shi for bringing this to our attention.} for all~$r=(r_1,\ldots,r_d)\in[0,1]^d$ by
\begin{eqnarray*}
    L^B(r) & := &  \left( \, \int_0^{r_j} Q^{B,j}(u) \, du \, \right)_{j=1}^d,
\end{eqnarray*}
where, for each~$j=1,\ldots,d$, and~$Q^{B,j}$ is the quantile function associated with the random variable~$\Sigma_{k=1}^d (\tilde X_j+\tilde X_k)/2d$. If the latter were replaced by~$\tilde X_j$, the Lorenz map would be the vector of marginal Lorenz curves. Mixing with the average across allocation introduces sensitivity to dependence between the marginal allocations. Note however, that~$Q^B:=(Q^{B,j})_{j=1}^d$ is not a valid vector quantile for~$\tilde X$, since $Q^B(U)$ is not distributed like~$\tilde X$ when~$U$ is uniform on~$[0,1]^d$. As a result, $L^B(r)$ is not a vector of resource shares, as is the case for the univariate Lorenz curve.

\subsubsection{Properties of the Lorenz map and Inverse Lorenz Function}
\label{sec:props}
The proposed multivariate extensions of the Lorenz curve in both \citeauthor{Taguchi:72a} (\citeyear{Taguchi:72a},\citeyear{Taguchi:72b}) and  \cite{koshevoy:1996} relate population proportions to a vector of resource shares. Our proposal differs substantially from these in that it directly relates a specific subset of the population, namely individuals with multivariate rank below~$r$ to their share of both resources. Beyond this major conceptual difference, we now investigate properties of our multivariate extension of the Lorenz curve that make it a valuable contribution.

\begin{enumerate}

\item Interpretation. Unlike other multivariate proposals, the Lorenz map shares the interpretation of the traditional Lorenz curve as the cumulative share of resources held by the lowest ranked individuals.

\item Computation. As shown in section~\ref{sec:comp}, the Lorenz map can be efficiently computed via convex programming. As an integrated vector quantile, it relies on the growing literature on computational geometry and computational optimal transport, where algorithms and implementations abound and are tested in a variety of applied fields.
This is in sharp contrast with the Lorenz zonoid proposed by \cite{koshevoy:1996} which is notoriously difficult to compute.

\item Statistical inference. As an integrated quantile, the Lorenz map is amenable to statistical inference. The convergence of sample analogues of vector quantiles to their theoretical counterpart was shown in \cite{CGHH:2017} and \cite{Figalli:2018}. Vector ranks are distribution free and can be used in rank based statistical procedures that emulate scalar rank-based inference, as shown in \cite{DS:2023}, \cite{ghosal:2021} and \cite{SDH:2023}. See the survey in \cite{Hallin:2022} and references within.

\item Uniqueness. The Lorenz map characterizes the distribution of the allocation it is associated with. This property is shared with the Lorenz zonoid of \cite{koshevoy:1996}, but not the other alternative proposals in the literature.

\begin{proposition}
\label{prop:char} The Lorenz map~$\Lmap_X$ characterizes the distribution of~$X$ in the sense that~$X$ and~$\tilde{X}$ are identically distributed if and only if~$\Lmap_X=\Lmap_{\tilde{X}}$.
\end{proposition}

\item Lorenz curve as a CDF. The Lorenz map is a map from~$[0,1]^d$ to~$[0,1]^d$. Hence, unlike the traditional scalar Lorenz curve, it cannot be a CDF. However, the Inverse Lorenz Function is the cumulative distribution function of a random vector on~$[0,1]^d$ by construction. This property is not shared by the alternative proposals in the literature.

\item Decomposition under independent attributes. As shown in example~\ref{ex:ind}, the Lorenz map reduces to a simple function of the marginal Lorenz curves in case marginal attribute allocations are independent. This feature is shared with the multivariate Lorenz proposal in \citeauthor{arnold:1983} (\citeyear{arnold:1983},\citeyear{arnold:1987}) but not the alternative proposals.

\item Dominance of egalitarian allocations.
In the univariate case, the Lorenz curve of the identical allocation~$Y=1$ almost surely, is~$L_Y(q)=q$, which is sometimes called the egalitarian line. The Lorenz curve of any other allocation~$Y\geq0$
is below the egalitarian line, i.e.,~$L_Y(q)\leq q$, for all~$q\in[0,1]$. For~$d>1$, the identical allocation of example~\ref{ex:egal} is a direct extension of the univariate notion of egalitarian. We show here that the Lorenz map and Inverse Lorenz Function of the identical allocation provide similar bounds in the multi-attribute case.  For this, we require allocations with components that display a form of positive association defined in assumption~\ref{+regdep}. 

\begin{assumption}
\label{+regdep} The vector quantile~$Q_{\tilde X}:=(Q_1,\ldots,Q_d)$ of~$\tilde X$ is such that, for each~$j$,
$\mathbb{E}\left[ Q_j(U_{1},\ldots,U_{d}) \, \vert \, U_{k}=u_{k},\mbox{all }k\ne j \right]$
is monotonically increasing in each of the~$u_{k}$, $k\ne j$, where the vector~$(U_1,\ldots,U_d)$ is uniform on~$[0,1]^d$.
\end{assumption}
This assumption imposes a type of positive dependence between the components of~$X$ through their ranks. More precisely, assumption~\ref{+regdep} imposes a form of \emph{positive regression dependence}, as in \cite%
{lehmann:1966}, between one resource and the others' ranks.
For allocations satisfying assumption~\ref{+regdep}, we show that Lorenz map and Inverse Lorenz Function of the identical allocation serve as upper and lower bounds, respectively.
\begin{proposition}
\label{prop:+regdep}
The Lorenz map of any allocation~$X$ satisfying assumption~\ref{+regdep} is
component-wise dominated by the Lorenz map of the identical allocation in example~\ref{ex:egal}.
Moreover, the Inverse Lorenz Function of allocation~$X$ is bounded below by the Inverse Lorenz Function of the identical allocation.
\end{proposition}

We argue in appendix~\ref{sec:egal} that 
defining egalitarianism solely by identical allocations is too restrictive in the case of multiple resources. In case~$d=2$, we show that a much larger class of allocations have Lorenz maps dominated by an egalitarian allocation from definition~\ref{def:egal}, which includes the identical allocation.
\end{enumerate}


\section{Multi-attribute inequality comparisons}
\label{sec:ineqorder}

We can use the vector Lorenz map~$\mathcal L_X$ of an allocation~$X$ introduced in section~\ref{sec:lmaps} as a tool to compare inequality of different allocations. We base an inequality dominance criterion to compare different allocations on the dominance of Lorenz maps. We develop a visualization tool for inequality dominance, and an inequality index for the cases, where the allocations are not Lorenz ordered.

\subsection{Lorenz dominance}

Consider two allocations~$X$ and~$X^\prime$, with respective Lorenz maps~$\Lmap_X$ and~$\Lmap_{X^\prime}$. If~$\Lmap_X(r)\geq \Lmap_{X^\prime}(r)$ for some vector rank~$r$, the same proportion of the population with vector ranks below~$r$ commands a larger share of all resources in allocation~$X$ than in allocation~$X^\prime$. If this is true for any vector rank~$r$ in~$[0,1]^d$, then, we say that allocation~$X^\prime$ is more unequal than allocation~$X$.

\begin{definition}
\label{def:s-order}
An allocation~$X^\prime$ is said to be more unequal in the Lorenz order than an allocation~$X$ if~$\Lmap_X(r )\geq \Lmap_{X^\prime}(r)$ for all~$r\in[0,1]^d$. We denote this~$X\succcurlyeq _{\Lmap}X^\prime$.\footnote{As a partial ordering based on cumulative sums of vector quantiles, the relation~$X\succcurlyeq_{\Lmap} X^\prime$ is a multivariate extension of the concept of majorization of \cite{hardy:1934}. It is different from existing multivariate notions of majorization reviewed in~ \cite{marshall:1979} and \cite{arnold:2018:book}, in that it relies on a multivariate reordering of the random vector allocation.}
\end{definition}

The Lorenz partial order of definition~\ref{def:s-order} is an implementable dominance criterion: The Lorenz maps can be computed and compared. The relation~$X\succcurlyeq_{\Lmap} X^\prime$ is equivalent to stochastic dominance of the random vector~$\Lmap_X(U)$, with~$U\sim U[0,1]^d$, over~$\Lmap_{X^\prime}(U)$ (see Section~3.8 of \cite{MS:2002}). Hence, dominance tests can be derived on the basis of sample analogues of the Lorenz maps to emulate the large literature on inference techniques to compare inequality of distributions of a single attribute. See \cite{DD:2000} and references within. 

Following the literature on the measurement of inequality, we assess the value of this implementable dominance criterion for inequality comparisons in two ways. First, we analyze the class of social evaluation functionals that are compatible with the Lorenz order, and show that they are rank-dependent social evaluation functionals, with weights decreasing in rank. Second, we identify the class of transfers that increase inequality as defined by this Lorenz criterion.

\subsubsection{Rank-dependent social evaluation functionals}
\label{sec:eval}

The first way to gain insight into the relevance of our multivariate Lorenz dominance criterion is to characterize the set of social evaluation functionals that are compatible with it. A social evaluation functional is a map~$S$ from an allocation~$X$, i.e., a random vector in~$\mathbb R_+^d$, to~$\mathbb R$, which orders allocations in their social desirability. A social evaluation functional~$S$ is compatible with the dominance criterion if~$X\succcurlyeq _{\Lmap}X^\prime\Rightarrow S(X)\geq S(X^\prime)$. Compatibility with Lorenz dominance is a form of inequality aversion of the social evaluation functional, since more equal allocations are deemed socially more desirable.

By construction, a social evaluation functional that is compatible with the Lorenz dominance order must satisfy anonymity and ratio-scale invariance. Anonymity, also called law-invariance or symmetry in the literature, refers to the fact that~$S(X)=S(X^\prime)$ whenever~$X$ and~$X^\prime$ are identically distributed. The identity of individuals does not matter in the social evaluation, so that a permutation of individuals in the population leaves~$S$ unchanged. Ratio-scale invariance refers to the fact that~$S(\lambda^\top X)=S(X)$ for any positive vector~$\lambda$. Hence, the social evaluation is not affected by a change in units of measurement.

Next, and more substantively, all social evaluation functionals that are compatible with the Lorenz dominance criterion are rank-dependent social evaluation functionals. Individuals are weighted in the social evaluation according to their rank in the distribution. To define and formalize this statement, start with the case of a single attribute. \cite{weymark:1981} shows that social evaluations that satisfy the comonotonic independence property defined below take the form of weighted sums of quantiles.

\noindent{\bf Property CI} (Comonotonic Independence).
A social evaluation functional~$S$ is said to satisfy comonotonic independence if, whenever~$X$, $X^\prime$ and~$Z$ are comonotonic allocations, and~$S(X) \geq S(X^\prime)$, then, for all~$\mu\in(0,1)$,
$S(\mu X+(1-\mu)Z) \geq S(\mu X^\prime+(1-\mu)Z)$.

In the univariate case, two allocations are called comonotonic if one is a positive increasing function of the other. In other words, individuals are ranked identically in both allocations. Comonotonic independence means that the comparison of two allocations with a common component only depends on the comparison between the two variable components, as long as rankings stay unchanged. As an illustration, when assessing the effect on household income distributions of a policy that only affects women, under perfect assortative matching, one need only look at the change in the distribution of women's income.

The same property of comonotonic independence can be entertained in the multi-attribute case, with the same interpretation. Two allocations are comonotonic if individuals are ranked identically in both allocations. Now, in case of random vectors~$X$ and~$X^\prime$, comonotonicity is defined in the same way by the fact that~$X$ and~$X^\prime$ have the same vector ranks. Definition~\ref{def:CO} below follows \citet{GH:2012} and \citet{EGH:2012}, where
it is called~$\mu$-comonotonicity\footnote{See \cite{PS:2010} for a discussion of this and other multivariate comonotonicity concepts.} of~$\tilde X^1,\ldots,\tilde X^J$.

\begin{definition}[Vector comonotonicity]
\label{def:CO} Random vectors~$X^1,\ldots,X^J$ on~$\mathbb{R}^d$ are
said to be \emph{comonotonic} if there exists a uniform random vector~$U$ on~%
$[0,1]^d$ such that~$\tilde X^j=Q_{\tilde X^j}(U)$, $j=1,\ldots,J$, almost surely, where~$Q_X$ is the
vector quantile of definition~\ref{def:VQ} associated with the distribution of~$X$, and~$\tilde X$ is the normalized version of~$X$ as in definition~\ref{def:Lorenz}.
\end{definition}

With this definition of comonotonicity (which coincides with the usual definition in the single attribute case), comonotonic independence of a social evaluation functional is still defined as property CI. 
If individuals are ranked identically in allocations~$X$ and~$X^\prime$, and~$X^\prime$ is socially less desirable than~$X$, then, adding to both~$X$ and~$X^\prime$ a third common allocation~$Z$ cannot reverse the ordering, if~$Z$ ranks individuals as~$X$ and~$X^\prime$ do.

As in \cite{weymark:1981} for the single attribute case, \cite{GH:2012} show that comonotonic additive social evaluation functionals are rank dependent, i.e., of the form
\begin{eqnarray}
\label{eq:rank}
    S_\phi(X) & := & \int_{[0,1]^d} \phi(u)^\top Q_{\tilde X}(u)du,
\end{eqnarray}
for some function~$\phi:[0,1]^d\rightarrow\mathbb R_+^d$. To each vector rank~$u$, $\phi$ associates the attribute-specific weights of ranked~$u$ individual in the social evaluation. We show that social evaluation functionals are only compatible with the Lorenz dominance order if they satisfy comonotonic additivity, hence if they are rank-dependent social evaluation functionals. There remains to determine which functions~$\phi$ make social evaluation functional~$S_\phi$ compatible with the Lorenz dominance criterion. As we discuss below, they are characterized by inequality aversion.

Inequality aversion of a rank dependent social evaluation functional implies a weighting scheme that gives more weight to lower ranked individuals. In the scalar case analyzed in \cite{weymark:1981}, an inequality averse rank dependent social evaluation functional is characterized by decreasing weights as ranks increase. We show a similar result in the multivariate case. Social evaluation functionals that are compatible with the Lorenz order of definition~\ref{def:s-order} are rank dependent social evaluation functionals with rank-specific weights of the form
\begin{eqnarray}
\label{eq:weights}
\phi_m(u) := \left( \int_{[0,1]^d} \mathds 1\{u\leq r\} \, dm_1(r),\ldots,\int_{[0,1]^d}\mathds 1\{u\leq r\} \, dm_d(r) \right)^\top,
\end{eqnarray}
where~$m_j$ is a non negative measure on~$[0,1]^d$, all~$j\leq d$. 

\begin{proposition}
\label{prop:eval}
    A social evaluation functional is compatible with the Lorenz dominance order of definition~\ref{def:s-order} if and only if it is of the form
    \begin{eqnarray}
    \label{eq:eval}
       S(X):=\int_{[0,1]^d}\phi_m(u)^\top Q_{\tilde X}(u) \, du
    \end{eqnarray}
\end{proposition}

A special case of weighting scheme satisfying proposition~\ref{prop:eval} is the case~$m_j=\delta_r$ all~$j$, where all individuals below rank~$r$ receive weight~$1$ and all individuals above rank~$r$ receive weight~$0$. More generally, individuals can be given different weights for different resource dimensions, but as non negative mixtures of the indicators~$\mathds 1\{u\leq r\}$, the weights are always decreasing in ranks.

\subsubsection{Increasing marginal inequality and increasing correlation}
\label{sec:MRT}

The second way we evaluate our Lorenz dominance criterion is by identifying transfers of resources between individuals that increase inequality according to this criterion. Since inequality is a cardinal aspect of the distribution, we consider a class of transfers that preserves the multivariate ranks. The transfers we consider are functions~$T:[0,1]^d\rightarrow\mathbb R^d$. If the~$j$-th component of transfer~$T$ is positive (resp. negative), it is added to (subtracted from) the endowment in resource~$j$ of individual with rank~$u$. 

\begin{definition}
    A rank preserving transfer from allocation~$X$ to allocation~$X^\prime$ is a transfer such that pre-transfer and post-transfer allocations are comonotonic (individuals preserve the same rank). Equivalently, it is a function~$T:[0,1]^d\rightarrow\mathbb R^d$ such that~$Q_{\tilde X^\prime}(u)=Q_{\tilde X}(u)+T(u)$ for all~$u\in[0,1]^d$, where~$Q_X$ is the
vector quantile of definition~\ref{def:VQ} associated with the distribution of~$X$, and~$\tilde X$ is the component-wise demeaned version of~$X$.
\end{definition}

First we show that the transfers that increase inequality according to the Lorenz criterion are the arbitrary combinations of rank preserving transfers of a non negative quantity of one of the resources from an individual with rank~$u_1$ to an individual with rank~$u_2\geq u_1$. 

\begin{proposition}
\label{prop:trans}
    An allocation~$X^\prime$ is more unequal than~$X$, i.e.,~$X^\prime\preccurlyeq_{\mathcal L}X$, if and only if an allocation with the same distribution as~$X^\prime$ can be obtained from~$X$ via an arbitrary sequence of rank preserving transfers~$T$ such that for all $r\in [0,1]^d$,
    \begin{eqnarray}
    \label{eq:trans}
        \int_{[0,1]^d}T(u)\mathds 1\{u\leq r\}\,du\leq0.
    \end{eqnarray}
\end{proposition}

The inequality in~(\ref{eq:trans}) expresses the fact that mass is transferred from lower ranked to higher ranked individuals.

A desirable feature of the Lorenz inequality ordering of definition~\ref{def:s-order} is its ability to rank two allocations~$X$ and~$X^\prime$, when the latter is obtained from the former through a transfer that increases inequality of the marginals or that increases the degree of positive dependence between the marginals. We formalize this feature with a specific type of multivariate transfer we call {\em Monotone Regressive Transfers}. We specialize the discussion to bivariate allocations to avoid wading into concepts of increasing multivariate dependence when~$d>2$.

\begin{definition}[Monotone Regressive Transfer, MRT]
\label{def:MRT}
A transfer~$T:[0,1]^2\rightarrow\mathbb R^2$ is a {\em monotone regressive transfer} if~$T$ is rank preserving and has non-negative Jacobian (i.e., the Jacobian's entries are all non negative).
\end{definition}

In the univariate case, a monotone regressive transfer reduces to a monotone mean preserving spread (\cite{Quiggin:92}), also called Bickel-Lehmann increase in dispersion (\cite{BL:76}). In the multivariate case\footnote{A related extension in the theory of multivariate risks was proposed in \cite{CGH:2016}.}, monotone regressive transfers weakly increase both marginal inequality and positive dependence. The former happens because each component of the transfer has non negative own derivative, hence is increasing in each component of the rank. The latter happens because the transfer has non negative cross derivative, hence increases the degree of positive dependence between the two resources.

\begin{proposition}[Monotonicity in MRT]
\label{prop:MRT}
If an allocation~$X^\prime$ is obtained from an allocation~$X$ through a monotone regressive transfer, then~$X\succcurlyeq_{\mathcal L}X^\prime$, i.e.,~$X^\prime$ is more unequal than~$X$ as defined by the Lorenz dominance partial order of definition~\ref{def:s-order}.
\end{proposition}

Proposition~\ref{prop:MRT} shows that the multivariate Lorenz dominance order of definition~\ref{def:s-order} therefore ranks an allocation as more unequal if the marginal resource allocations are weakly more unequal and if the marginal resource allocations are weakly more positively dependent. This is in contrast with the Lorenz dominance order based on inclusion of Lorenz zonoids proposed in \cite{koshevoy:1996}. Indeed, by Proposition~8 in \cite{koshevoy:2007}, if two allocations have identical marginals, and~$X$ dominates~$X^\prime$ in the Lorenz dominance order based on Lorenz zonoid inclusion, then~$X$ and~$X^\prime$ are identically distributed.

The Lorenz dominance ordering of definition~\ref{def:s-order} does not satisfy the {\em uniform majorization principle} proposed by \cite{kolm:1977}. In case of discrete populations, the uniform majorization principle stipulates that inequality should be reduced through multiplication by a doubly stochastic matrix different from a permutation. However, as \cite{Dardanoni:93} points out, such transformations can increase correlation and therefore increase inequality in an egregious way. See the discussion in \cite{Savaglio:2006}. We show that a similar issue arises with the continuous version of the uniform majorization principle. The latter requires an inequality dominance order to be monotonic with respect to the concave order. From theorem~4 of \cite{GH:2012}, we deduce that a social evaluation functional that satisfies uniform majorization and comonotonic independence must be equal to
\begin{eqnarray}
\label{eq:UM-Gini}
    S_{UM}(X):=1-\int_{[0,1]^d}u^\top Q_{\tilde X}(u)du
\end{eqnarray} 
up to an affine transformation. We show in appendix~\ref{app:UM} that in the case of bivariate allocation~$X=(X_1,X_2)$, $S_{UM}$ is minimized when the two components~$X_1$ and~$X_2$ of allocation~$X$ are independent, which runs against the intuition that increased dependence can increase inequality\footnote{Note that we are considering inequality over outcomes, not welfare inequality. Hence, the point made in \cite{AB:1982}, that increased correlation may decrease utilitarian welfare inequality when resources are complements, doesn't apply here.}.


\subsection{Visualization of Lorenz dominance}

Failures of Lorenz dominance can be visualized with the Inverse Lorenz Function. Consider two allocations~$X$ and~$X^\prime$, with respective Inverse Lorenz Functions~$l_X$ and~$l_{X^\prime}$.
If~$l_X(z)\leq l_{X^\prime}(z)$ for some vector of shares~$z$, a larger proportion of the population commands the same share of resources in allocation~$X^\prime$ than in allocation~$X$. Lorenz dominance of allocation~$X$ over allocation~$X^\prime$ (in the sense of definition~\ref{def:s-order}) implies that the relation ~$l_X(z)\leq l_{X^\prime}(z)$ holds for each resource share vector~$z$ (see proposition~\ref{prop:order} in the appendix).

In case of bivariate allocations, the latter can be easily visualized on~$[0,1]^2$ through the relative positions of the level sets of the Inverse Lorenz Function, which we call~$\alpha$-Lorenz curves, denoted
\begin{eqnarray*}
    \mathcal C_X^\alpha:=\{z\in[0,1]^2: l_X(z)=\alpha\}, \mbox{ for each }\alpha\in(0,1).
\end{eqnarray*}
The $\alpha$-Lorenz curves provide a visualization of Lorenz dominance. We can compare the inequality of different allocations based on the shape and relative positions of their respective~$\alpha$-Lorenz curves. 
Suppose~$X$ is less unequal in the Lorenz dominance order than~$X^\prime$. Then, by proposition~\ref{prop:order} in appendix~\ref{sec:proofs}, for any~$z\in\mathcal C^\alpha_{X^\prime}$, $l_{X}(z)\leq l_{X^\prime}(z)=\alpha$. So~$z\in\mathcal C^{\tilde \alpha}_{X}$ with~$\tilde \alpha\leq \alpha$. This can be visualized as a shift to the north-east of the~$\alpha$-Lorenz curves of the more unequal 
allocation~$X^\prime$ to the~$\alpha$-Lorenz curves of the less unequal allocation~$X$.

Figure~\ref{fig:ex} and~\ref{fig:lognormal_ordering} display $\alpha$-Lorenz curves of the multivariate lognormal allocation~$X$ defined by 
\begin{eqnarray}\label{ex:mlnorm}
    \ln X \sim
    \mathcal{N}\left(
    \begin{bmatrix}
    0 \\ 0
    \end{bmatrix}
    ,
    \begin{bmatrix}
        \sigma_1^2 & \rho \\
        \rho & \sigma_2^2
    \end{bmatrix}
    \right)
\end{eqnarray}
where $\mathcal{N}$ is the normal distribution, $\sigma_1,\sigma_2 > 0$ control the dispersion of the respective marginals and $\rho$ controls the degree of dependence. 

In figure~\ref{fig:ex}, we revisit the three special cases of identical allocations, independent attributes and comonotonic attributes. We also add the counter-monotonic case, where individuals are ranked in opposite order for the two resources, as well as intermediate dependence cases. Figure~\ref{fig:ex} shows the~$\alpha$-Lorenz curve (with~$\alpha=0.9$) of the identical allocation compared to multivariate lognormal allocations with variances~$\sigma_1=\sigma_2=1$ and correlation coefficients~$\rho=-0.99$ (counter-monotonic), $\rho=-0.6,-0.3$, $\rho=0$ (independent), $\rho=0.3,0.6$, and~$\rho=0.99$ (comonotonic).

\begin{figure}[H]
\begin{center}
\includegraphics[scale=0.6]{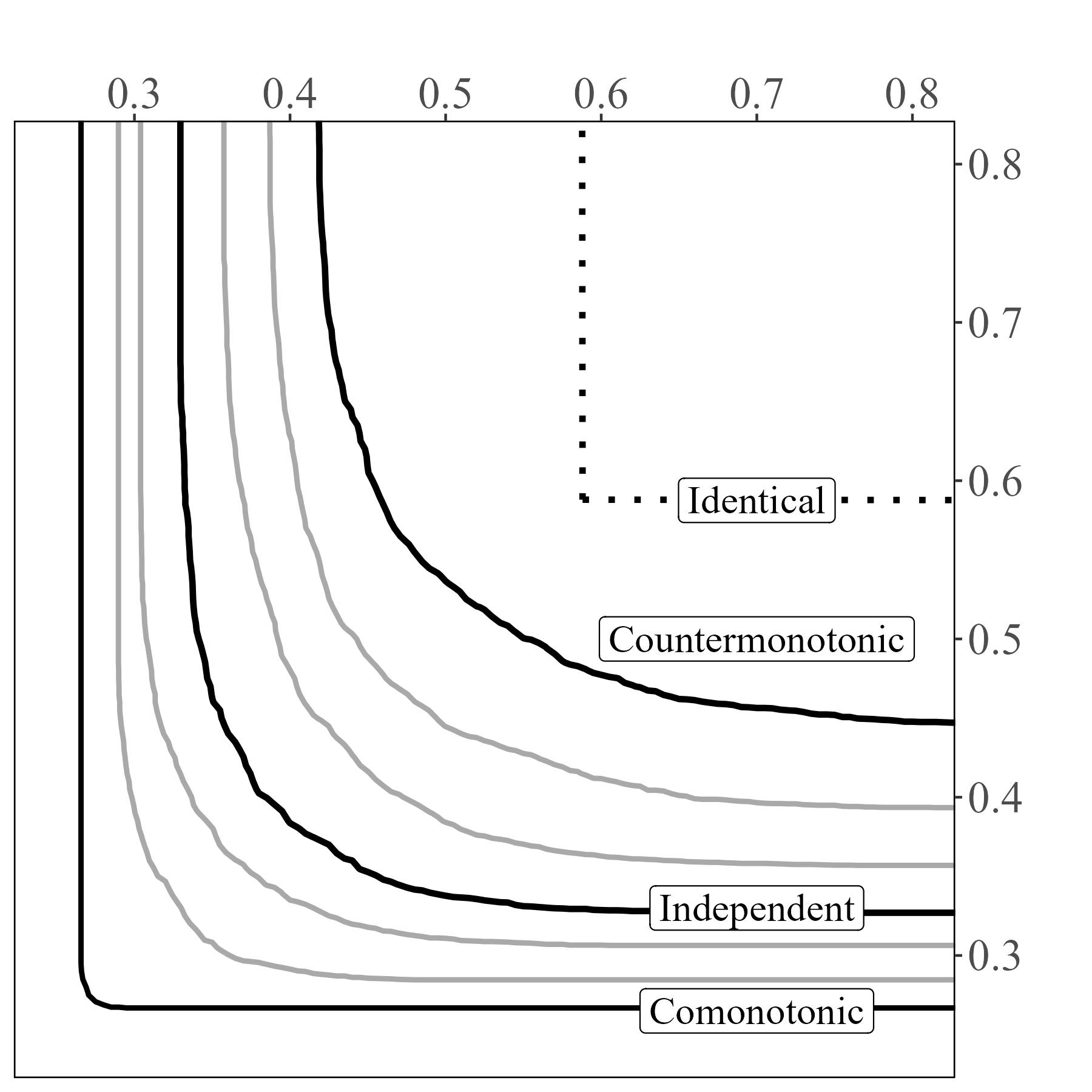}
\end{center}
\caption{$0.9$-Lorenz curves of multivariate lognormal random vectors that vary by correlation between marginals defined in \eqref{ex:mlnorm}. The scale parameters are $\sigma_1 = \sigma_2 = 1$. The curves correspond to different correlation coefficient $\rho$ (from bottom left to top right: $-0.99, 0.6, -0.3, 0, 0.3, 0.6, 0.99$). The nested curves show that increasing dependence between fixed marginals increases inequality.}
\label{fig:ex}
\end{figure}

Visually, inequality can be assessed by the departure of~$\alpha$-Lorenz curves from those of the identical allocation. This visual comparison is facilitated by the fact that they are shaped like indifference curves. 
In addition, correlation information is preserved through the curvature of the~$\alpha$-Lorenz curves, which decreases when positive dependence increases.

\begin{proposition}
\label{prop:curves}
(1) The~$\alpha$-Lorenz curves~$\mathcal C^\alpha_X$ are the level curves of a bivariate cdf, hence they are downward sloping, non decreasing in~$\alpha$ and they do not cross. In addition, (2) The~$\alpha$-Lorenz curves~$\mathcal C^\alpha_X$ are convex if
\begin{eqnarray*}
\frac{\partial l_X}{\partial z_2} \frac{\partial l_X^2}{\partial z_1\partial z_2} - \frac{\partial l_X}{\partial z_1} \frac{\partial l_X^2}{\partial z_2^2} \geq0.
\end{eqnarray*}
\end{proposition}

\begin{figure}[htbp]
\centering
\includegraphics[scale = 0.6]{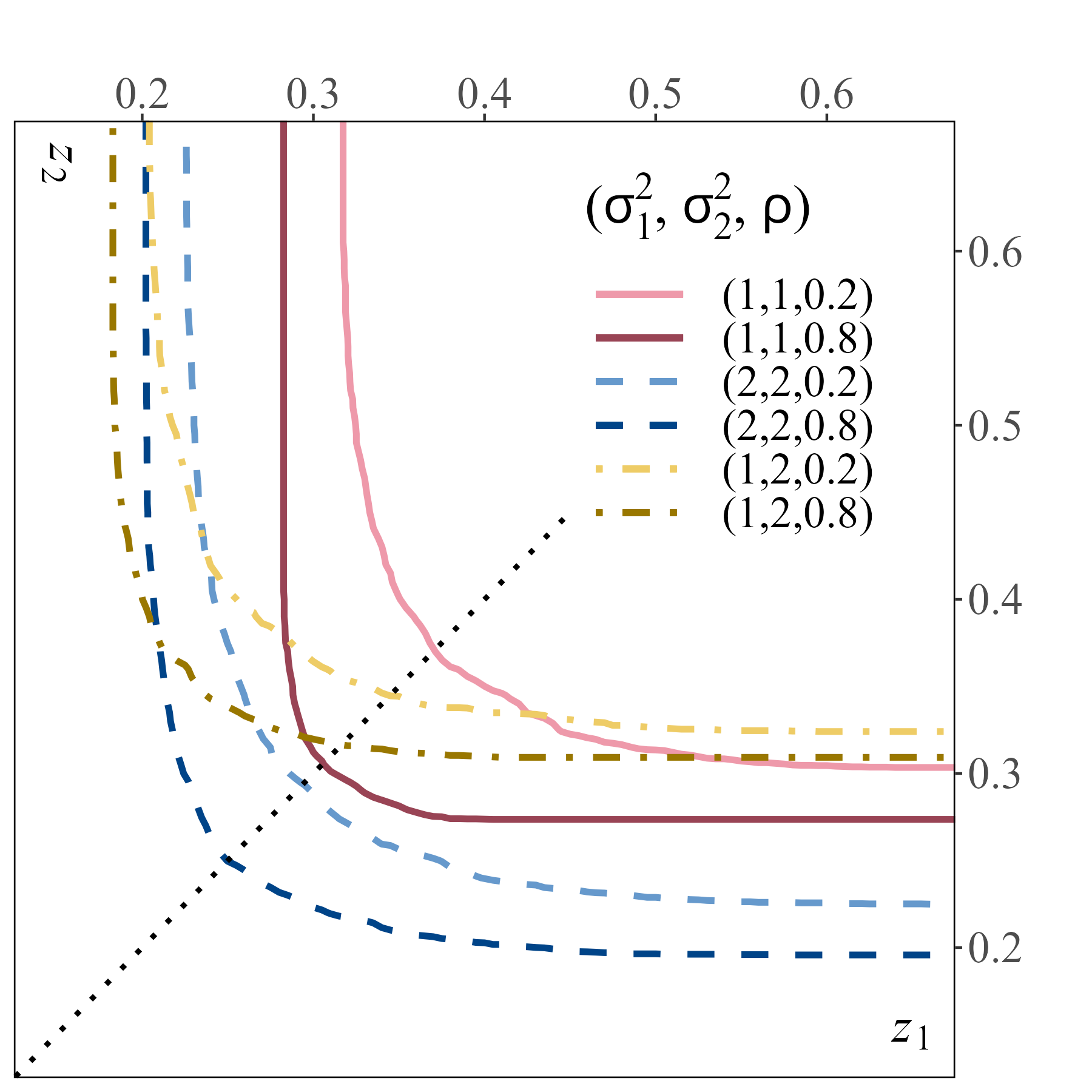} 
\caption{$0.9$-Lorenz curves of multivariate lognormal random vectors that vary by correlation and scale parameters defined in \eqref{ex:mlnorm}. Notably the dotted-dashed lines display when $\sigma_1\neq\sigma_2$ showing a skewness towards the axis $z_2$ of the marginal with larger value. The diagonal black dotted line mark where the identical allocation lies and helps to show skewness towards an axis.}
\label{fig:lognormal_ordering}
\end{figure}

Figure~\ref{fig:lognormal_ordering} compares~$\alpha$-Lorenz curves of~$6$ different allocations that are multivariate lognormally distributed as in \eqref{ex:mlnorm}, for~$\alpha=0.9$. The parameters~$\sigma_1^2$ and~$\sigma_2^2$ take values~$1$ or~$2$, whereas ~$\rho$ takes values~$0.2$ or~$0.8$. In case of marginals with different~$\sigma$, the asymmetry is reflected in the ~$\alpha$-Lorenz curves. Moreover, other things equal, inequality increases with~$\sigma$, which measures inequality in the marginals, and with~$\rho$, which measures correlation.

Finally, figure~\ref{fig:trade} shows an example of two multivariate lognormally distributed allocations $X$ and $X^\prime$ such that the marginals of $X$ are more unequal than those of $X^\prime$, but $X^\prime$ is more unequal overall due to positive dependence of its marginals. Specifically, the marginals of $X$ are independent and have the same scale parameter $\sigma^2 = 1.2$, while $X^\prime$ has marginals with correlation parameter $\rho = 0.9$ and scale $\sigma^2 = 1$ each. This shows that the Lorenz inequality dominance ordering does not imply dominance of the marginals, and how it can instead incorporate trade-offs between marginal inequality and positive dependence.

\begin{figure}[htbp]
\centering
\includegraphics[scale = 0.6]{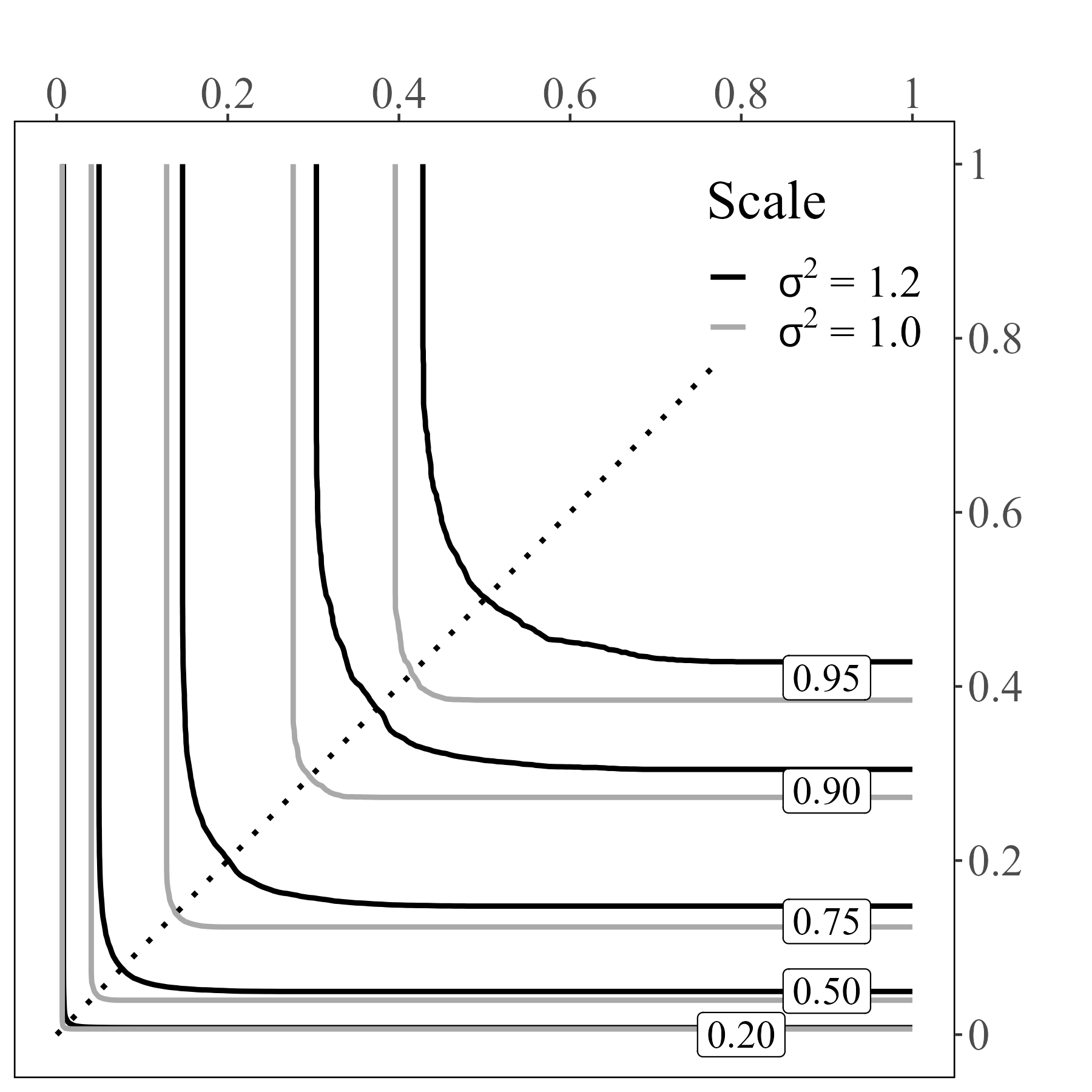}
\caption{$\alpha$-Lorenz curves of~$X$ with independent lognormal marginals with~$\sigma^2=1.2$ in black and~$X^\prime$ with dependent lognormal marginals with~$\sigma^2=1$ and correlation~$\rho=0.9$ in gray.}
\label{fig:trade}
\end{figure}


\subsection{Multivariate Gini inequality index}
\label{sec:gini}

The Lorenz dominance ordering is a partial ordering of multivariate distributions. In many cases, $\alpha$-Lorenz curves may cross. For a complete inequality ordering, we also propose an extension of the classical Gini index to compare inequality in multi-attribute allocations.
\cite{gajdos:2005} propose a multivariate Gini coefficient based on aggregation across individuals first, then across dimensions, which removes the effect of dependence across attributes. \cite{DL:2012} propose to aggregate across dimensions first, then across individuals, in order to keep track of correlation.
From the volume of the Lorenz zonoid, a multidimensional Gini coefficient can be derived naturally as in \cite{koshevoy:1997}. An alternative strategy is followed by \cite{arnold:1983}, \cite{koshevoy:1997}, who extend the definition based on the sum of all
distances between pairs of individuals.\footnote{Other multivariate Gini indices based on multivariate Lorenz curve proposals include \cite{banerjee:2010}, \cite{grothe:2021}, and \cite{SJ:2020}.}

The univariate Gini index can be interpreted as the average deviation from the egalitarian allocation, the univariate version of our identical allocation. We emulate this interpretation and define a multivariate Gini based on an average deviation from the Lorenz map~$r\mapsto(r_1r_2\cdots r_d,\ldots,r_1r_2\cdots r_d)$. The deviation measure we propose is
\begin{eqnarray}
\label{eq:dev}
\displaystyle \sum_{j=1}^d \; \int_{[0,1]^d} \; \left[ \; \prod_{k=1}^dr_k-\Lmap_j(r) \right] dr,
\end{eqnarray}
where~$\Lmap_j$ is the~$j$-th component of the Lorenz map~$\Lmap_X$, with~$j=1,\ldots,d$. 
After normalization, (\ref{eq:dev}) becomes
\begin{eqnarray}\label{eq:GiniLmap}
G(X) & = &  1 - \frac{2^d}{d}\left( \; \int_{[0,1]^d} \;\sum_{j=1}^d\mathcal{L}_j(r)  \,dr \; \right),
\label{eq:Gini}
\end{eqnarray}
which yields the following definition.

\begin{definition}[Gini Index]
\label{def:Gini}
(\ref{eq:Gini}) defines the multivariate Gini index of allocation~$X$.
\end{definition}

The traditional Gini index of a univariate allocation can also be characterized as a weighted sum of outcomes, where the weights are increasing linearly in the rank of the individual in the population. Hence, the negative of the Gini, seen as a social evaluation functional displays inequality aversion by giving more weight to the outcomes of lower ranked individuals than to those of higher ranked ones. 

The same interpretation is valid for our multivariate Gini. Specifically, 
\begin{eqnarray}
    S(X) & := & \frac{d}{2^d} \, (1-G(X)) \notag\\
     & = & \int_{[0,1]^d} \left( \int_{[0,1]^d} \, \left( \mathds 1\{u\leq r\}, \ldots, \mathds 1\{u\leq r\} \right)^\top Q_{\tilde X}(u) \, du \right) \, dr \notag\\
    & = & \int_{[0,1]^d} \, \left( \prod_{j=1}^d (1-u_j), \ldots, \prod_{j=1}^d (1-u_j) \right)^\top Q_{\tilde X}(u) \, du \label{eq:giniVQ}
\end{eqnarray}
is an inequality averse social evaluation functional of the form~(\ref{eq:eval}) with uniform measures on~$[0,1]^d$ in~(\ref{eq:weights}). So~$S(X)$ evaluates the social desirability of an allocation with a weighted sum of individual endowments, where the weights~$\Pi_{j=1}^d(1-u_j)$ are component-wise decreasing in the individual's rank~$u=(u_1,\ldots,u_d)$. 

The Gini index correspondingly takes the form~$G(X)=1-(2^d/d)S(X),$ with~$S$ as in~(\ref{eq:giniVQ}).
For instance, in the case of bivariate allocations, (\ref{eq:Gini}) takes the form
\begin{eqnarray}
\label{eq:Gini-Social}
G(X)  & = & 2 \int_0^1\!\!\!\int_0^1 (u_1+u_2-u_1u_2) \left( Q_1(u)+Q_2(u)\right) du_1du_2 \, - \, 3.
\end{eqnarray}
In expression~(\ref{eq:Gini-Social}), $Q_1(u)$ and~$Q_2(u)$ are the components of~$Q_{\tilde X}$, so that~$Q_1(u)+Q_2(u)$ is the sum of the two normalized resource allocations of the individual in the population with vector rank~$(u_1,u_2)$. Hence, the Gini index is indeed a weighted sum of outcomes, with weights~$(u_1+u_2-u_1u_2)$ increasing with the vector ranks~$(u_1,u_2)$. It is a genuinely multivariate extension in that the weighting scheme, hence the social evaluation of inequality, depends on multivariate ranks of individuals.

\begin{continued}[Examples~\ref{ex:discrete} continued]
The Gini coefficient for discrete distributions can be computed by plugging in the Lorenz map from section~\ref{sec:comp} into~(\ref{eq:Gini}).
\end{continued}

\begin{continued}[Examples~\ref{ex:ind} and~\ref{ex:comon} continued]
We compare Gini indices in the independent case with the perfect comonotonicity case, where~$X_1$ and $X_2$ have the same marginal distributions (uniform on $[0,2]$). We verify (analytically for~$r_1+r_2<1$ and numerically using Wolfram for~$r_1+r_2\geq 1$) that~$Q_1(u)+Q_2(u)$ is smaller in the comonotonic case, than in the independent case. Hence the Gini index (and the measure of inequality) is larger in the comonotonic case.
\end{continued}

\begin{example}[Countermonotone Resources]
If we have $X_1+X_2=2$ a.s., then~$Q_1(u)+Q_2(u)=2$ for almost all~$u$,
and we obtain~$\Lmap_1(r_1,r_2) +\Lmap_2(r_1,r_2)=2r_1r_2\geq r_2L_1(r_1) +r_1L_2(r_2)$, so that, in particular, the Gini index in the countermonotone case is the same as in the case of the identical allocation, i.e., equal to~$0$, and both are smaller than the Gini of the allocation with independent resources. This is consistent with the fact that these allocations~$X$ are considered egalitarian according to definition~\ref{def:egal} in appendix~\ref{sec:egal}.
\end{example}

The Gini index of definition~\ref{def:Gini} is in~$[0,1]$ under assumption~\ref{+regdep}. It equals~$0$ for the identical allocation. It tends to~$1$, when the Lorenz map tends to~$0$ (extreme inequality). The Gini index of an allocation with independent components reduces to the average of classical scalar Ginis of both components. Like the classical scalar Gini index, it preserves the Lorenz inequality ordering, in the sense that higher inequality according to~$\preccurlyeq_{\Lmap}$ implies a larger value of the Gini index. In other words, $X\succcurlyeq_{\Lmap} X^\prime$ implies~$G(X)\leq G(X^\prime)$, so that the negative of the Gini is a compatible social evaluation functional. Hence it inherits the properties of anonymity, scale invariance and comonotonic independence.

\subsubsection*{Multivariate S-Gini}\label{subsect: S-Gini}

The multivariate Gini in expression~(\ref{eq:Gini}) is the suitably normalized negative of an inequality averse social evaluation functional of the form~(\ref{eq:eval}) with uniform measures on~$[0,1]^d$ in~(\ref{eq:weights}). It can be extended to reflect varying concern for inequality in different attributes. To achieve this, a multivariate Gini coefficient can be defined as~$1-cS(X)$, where~$c$ is a normalizing constant and~$S$ is a social evaluation functional that reflects different degrees of inequality aversion in different attributes.

In the univariate case, to reflect varying degrees of inequality aversion, \cite{DW:80} propose a single parameter family of Gini coefficients, called S-Gini, defined by
\begin{eqnarray*}
G_\delta(X) & := & 1-\delta(\delta-1) \int_{[0,1]}(1-r)^{\delta-2}\mathcal{L}(r)\, dr,
\end{eqnarray*}
where~$\mathcal L$ is the traditional Lorenz curve, and~$\delta$ ranges from~$1$, corresponding to indifference to inequality, to the Rawlesian extreme at the limit $\delta\rightarrow\infty$, where only the poorest individual matters\footnote{We were unable to locate a precise statement of this in the literature, so we include it in proposition~\ref{prop:S-Gini} in the appendix with a proof for completeness.}. 

The S-Gini family of \cite{DW:80} can be extended to the assessment of multivariate inequality within our framework. Let~$\delta=(\delta_1,\ldots,\delta_d)$ be a $d$-dimensional parameter, where~$\delta_j\in[1,\infty)$, $j=1,\ldots,d,$ reflects the concern for inequality in attribute~$j$. We define the family of multivariate S-Gini coefficients of inequality of an allocation~$X$ with Lorenz map~$\mathcal L_X=(\mathcal L_1,\ldots,\mathcal L_d)$ as 
\begin{eqnarray*}
\label{eq:S-Gini}
G_\delta(X)
& := & 1-c_\delta S_\delta(X),
\end{eqnarray*}
where~$c_\delta:=2^{d-1}/\sum_{j=1}^d\delta_j^{-1}$ is a normalizing constant, and~$S_\delta$ is the social evaluation functional
\begin{eqnarray*}
\label{eq:S-Gini Social}
S_\delta(X)
& := & \sum_{j=1}^d (\delta_j-1) \int_{[0,1]^d}(1-r_j)^{\delta_j-2}\mathcal{L}_j(r)\, dr.
\end{eqnarray*}
The normalizing constant~$c_\delta$ is chosen such that the multivariate S-Gini~$G_\delta$ lies in~$[0,1]$ and is zero in case of the identical allocation. There remains to verify that the social evaluation functional~$S_\delta$ is indeed of the form~(\ref{eq:eval}), and hence compatible with the Lorenz order. Indeed, we have
\begin{equation*}
    S_\delta(X)=\sum_{j=1}^d \int_{[0,1]^d}\mathcal{L}_j(r)\, dm^{(\delta)}_j(r),
\end{equation*} 
with~$m^{(\delta)}_j(r)=\left(\prod_{l=1,l\neq j}^d r_l\right)[1-(1-r_j)^{\delta_j-1}]$, for each~$j=1,\ldots,d$. The multivariate S-Gini thereby incorporates varying degrees of inequality aversion for different attributes. We recover the S-Gini of \cite{DW:80} when~$d=1$, and the multivariate Gini of section~\ref{sec:gini} when~$\delta_j=2$, for all~$j=1,\ldots,d$, as desired. We also recover Rawelsian limits as~$\delta_j$ tends to zero as formalized in proposition~\ref{prop: multivariate s-gini} in the appendix.


\section{Empirical Illustration}
\label{sec:appli}

In this section, we apply our methodology to the analysis of income-wealth inequality in the United States between~1989 and~2022, based on the public version of the triennial Survey of Consumer Finances (SCF). Wealth refers to all assets, financial and otherwise. Details of the sampling technique and a discussion of specific features and issues with the data set are given in appendix~\ref{app:data}. A guide for practical implementation of the computational procedure outlined in section~\ref{sec:comp} is given in appendix~\ref{sec:algo}. We refer to inequality displayed by our measure as overall inequality, while specific marginal inequality is described as wealth or income inequality.


\subsection{Income-wealth~$\boldsymbol{\alpha}$-Lorenz curves}

Figure~\ref{fig:alorenz_decades} shows the $\alpha$-Lorenz curves for $\alpha = 0.6,0.8,0.95$ for the years 1989, 2007, 2010, and 2022. There is a general worsening of overall inequality over~3 decades since the curves shift away from the north-east corner. The tight curvature also reflects the positive correlation of income and wealth as in figure~\ref{fig:ex}. Using figure~\ref{fig:lognormal_ordering} as reference, the skew towards the wealth axis indicates inequality from the wealth marginal is dominant at these $\alpha$-levels, as expected. 


\begin{figure}[htbp]
\centering 
    \phantom{aaaa}\includegraphics[scale = 0.6]{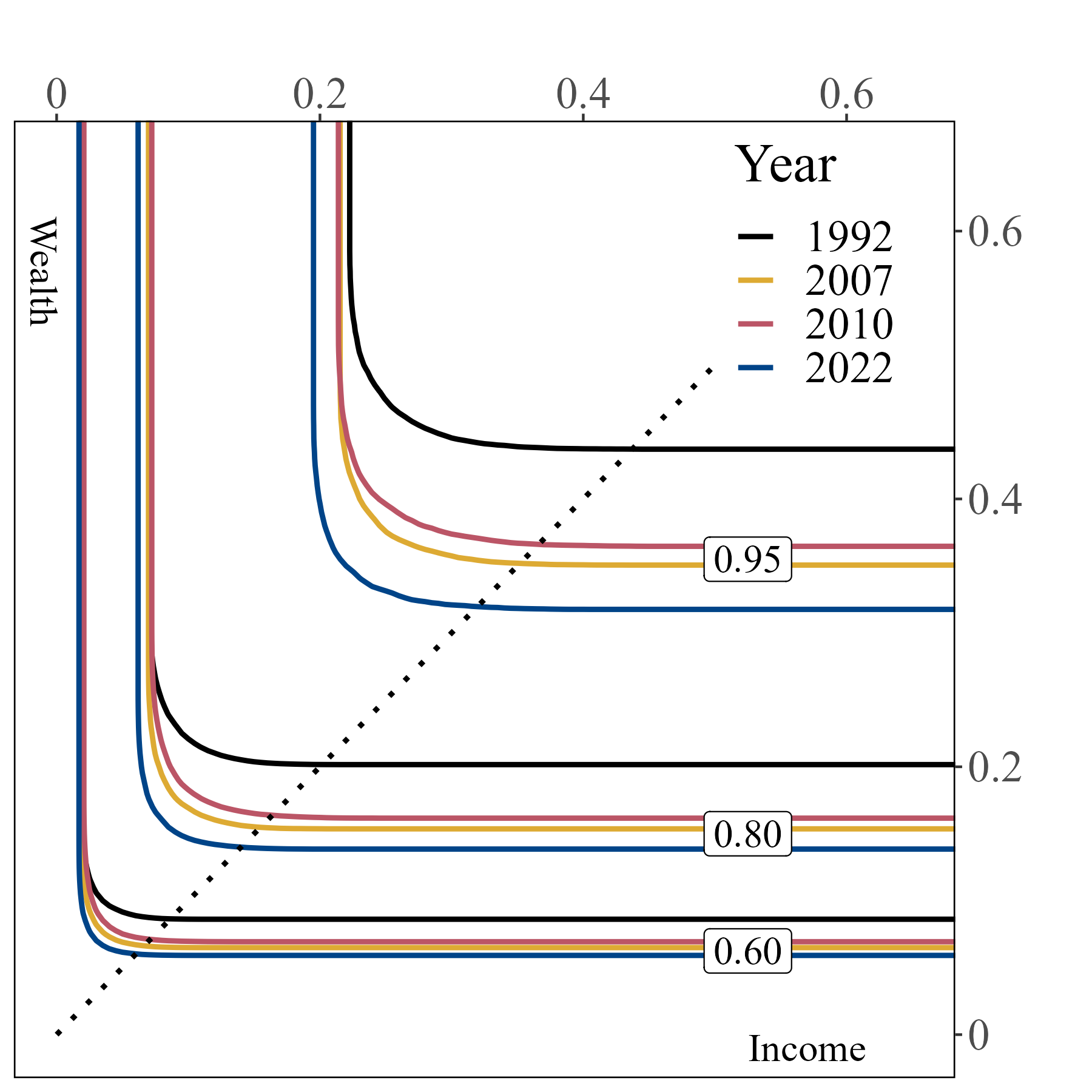}
\caption{$\alpha$-Lorenz curves ($\alpha=0.6,0.8,0.95$) for US Income-Wealth for years 1989, 2007, 2010, and 2022.}
\label{fig:alorenz_decades}
\end{figure}

\subsection{Resource shares}

Figure~\ref{fig:nonagg} shows resource shares of the~$25\%$ fraction of the the population below rank~$r=(0.5,0.5)$ as well as the~$90\%$ fraction\footnote{The fraction of the population is exactly~$0.95^2=0.9025$.} of the population below rank~$r=(0.95,0.95)$. The shares in both resources of the bottom $90\%$ have been steadily declining and the shares of wealth are lower than the shares of income. Between 2007 and 2010, we see income shares increasing relatively more than the decrease in wealth shares. Wealth and income shares both fell from 2010 to 2022, explaining the shift in the curves from figure ~\ref{fig:alorenz_decades}. As for the bottom $25\%$, changes over time are minor compared to those of the bottom $90\%$.

\begin{figure}[htbp]
\centering
\includegraphics[scale = 0.5]{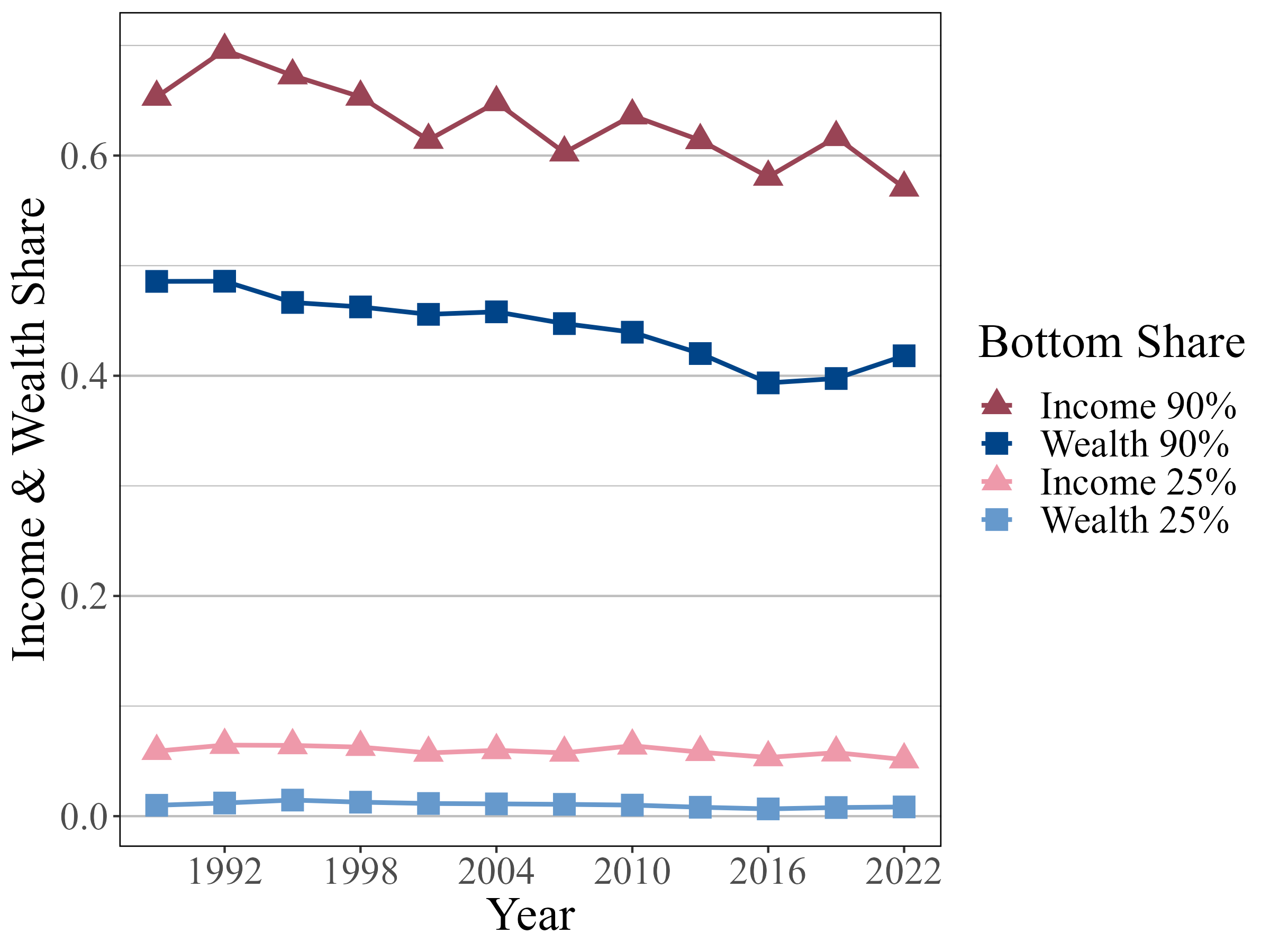}
\caption{The income and wealth shares of the~$90\%$ (resp. $25\%$) fraction of the population ranked below~$r = (0.95,0.95)$ (resp.~$(0.5,0.5)$).}
\label{fig:nonagg}
\end{figure}

\subsection{Gini indices}\label{sec:empGini}

Figure~\ref{fig:tau} displays the marginal Gini indices for income and for wealth, the multivariate Gini index based on \eqref{eq:GiniLmap}, as well as Kendall's $\tau$ for the dependence between income and wealth over time. The multivariate Gini shows a steady increase in overall inequality. If the resources were independent, the multivariate Gini would be the average of the marginal Ginis. In the present case, the multivariate Gini reveals a positive association between resources since it is higher than the average of the marginal Gini indices. The multivariate Gini shows reduced overall inequality between $2007$ and $2010$. The decreased correlation and income inequality may have been sufficient to offset the rise in wealth inequality.

\begin{figure}[htbp]
\centering
\includegraphics[scale = 0.5]{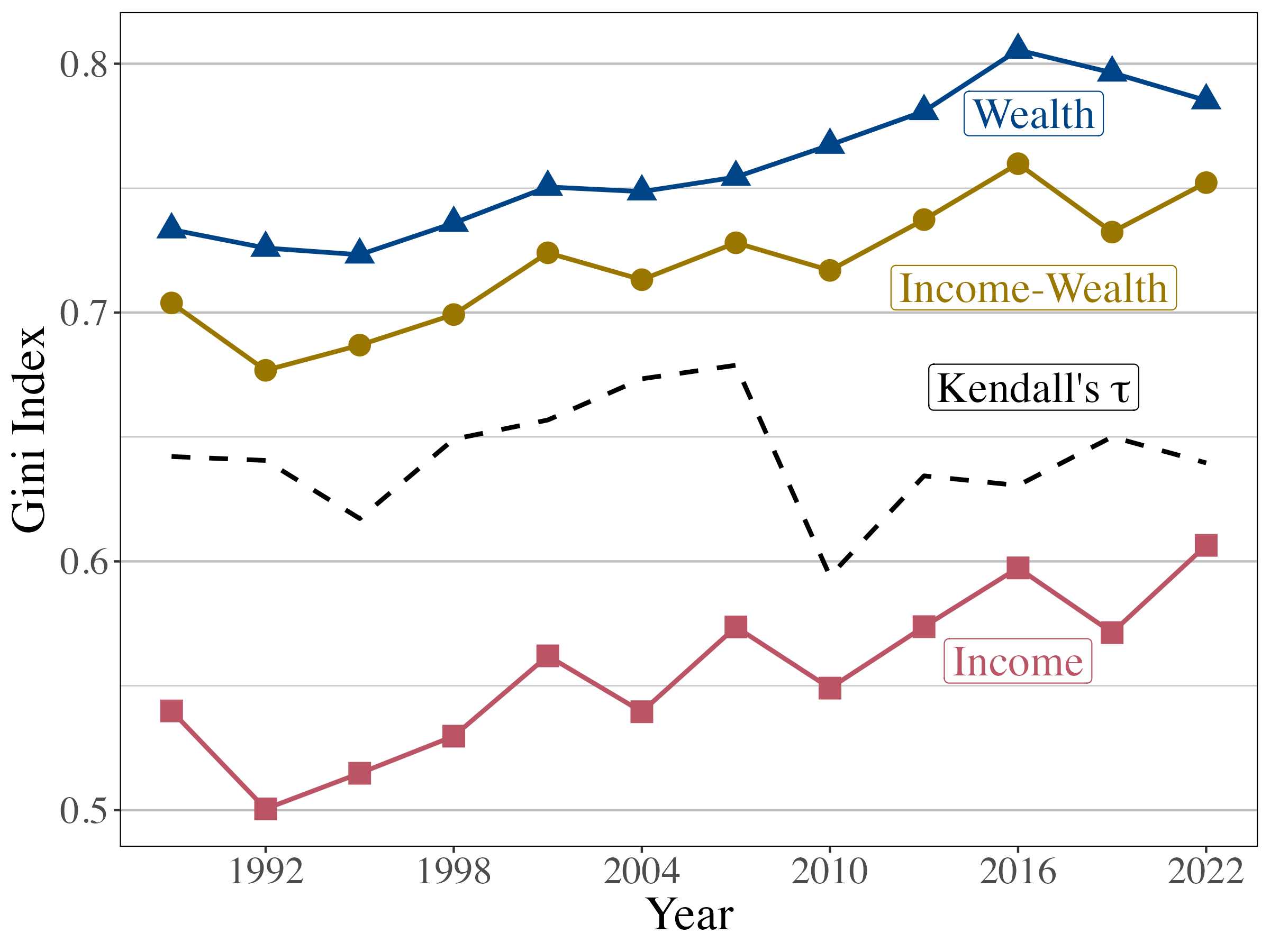}
\caption{Top: Gini indices for income and for wealth, multivariate Gini index, and Kendall's $\tau$ (dashed) for US Income-Wealth across 1989-2022.}
\label{fig:tau}
\end{figure}

Inequality analysis across groups can reveal further insights. Figure~\ref{fig:gini-race-age} shows Gini indices among White and Black respondents as well as among the working age (64 years and below) and retiring age populations (65 years and above). While overall inequality has worsened among White respondents, the inequality among Black respondents has remained steady. When comparing inequality across age groups, they both exhibit a steady increase in overall inequality, however the multivariate Gini among the retiring age group inherits the variability in the income marginal.

\begin{figure}[htbp]
\centering
\includegraphics[width = 0.45\textwidth]{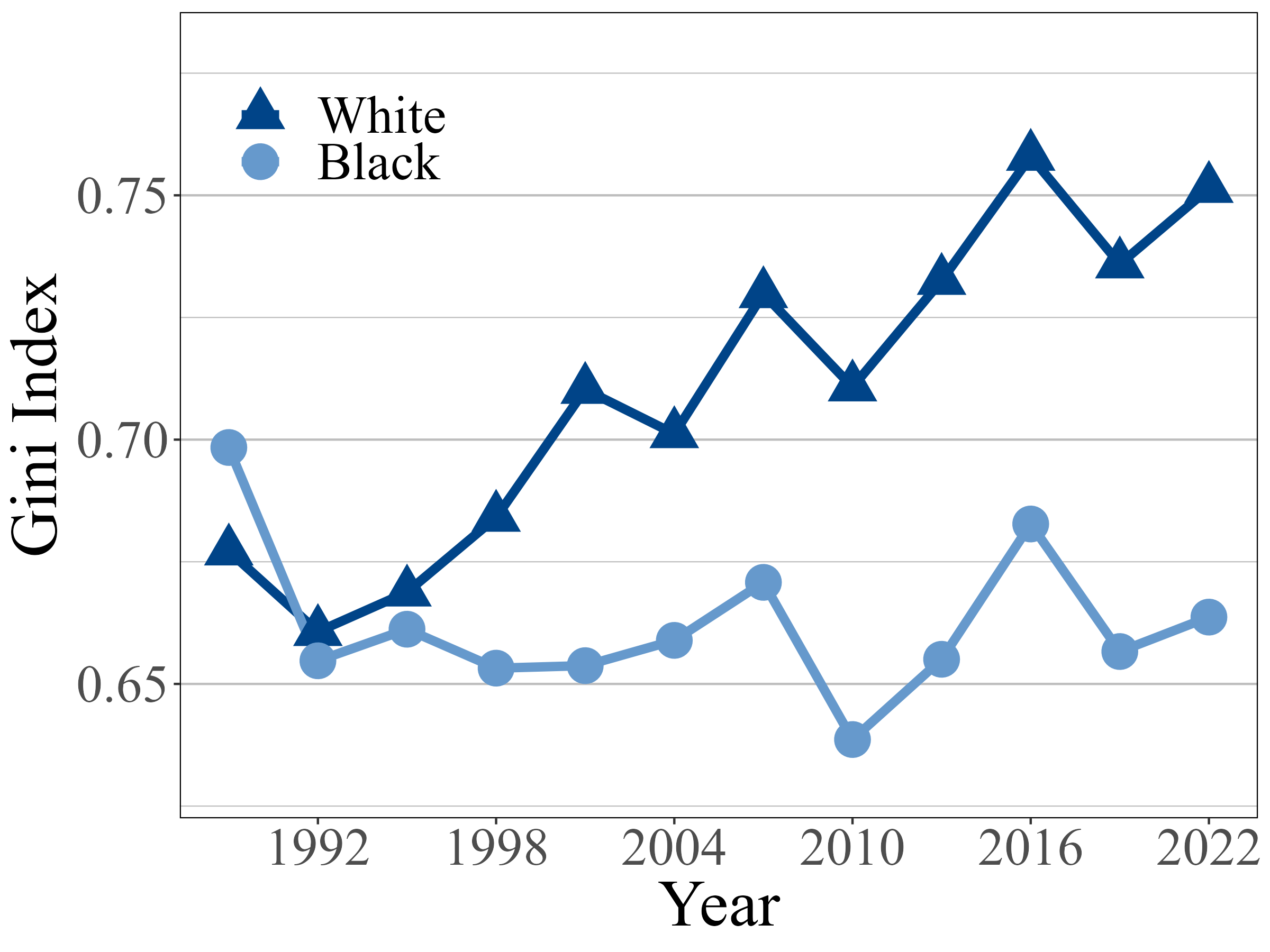}\includegraphics[width = 0.45\textwidth]{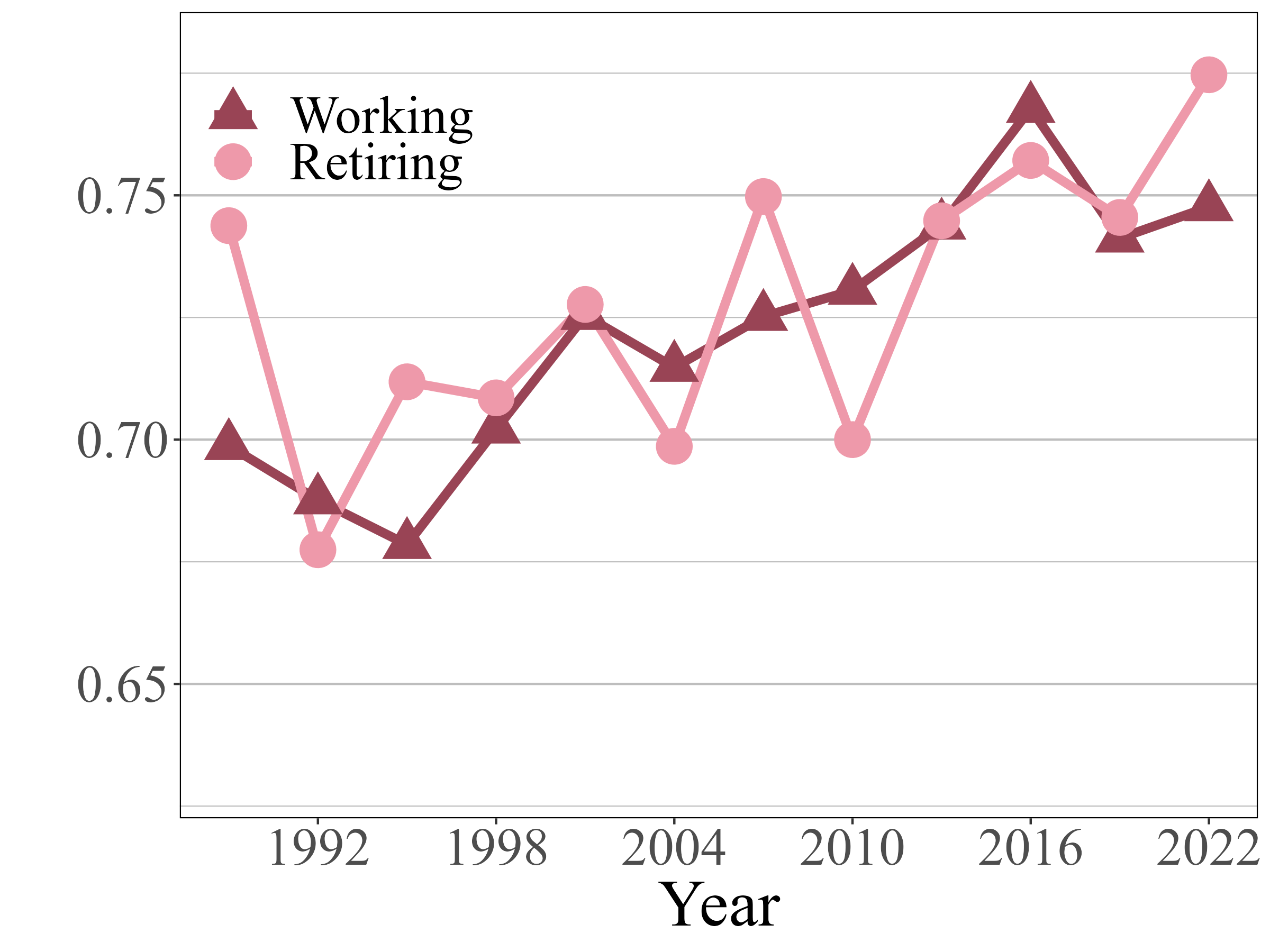}
\caption{Left: Multivariate Gini indices among White and Black respondents. A significant overall inequality gap has formed over time between the two segments. Right: Multivariate Gini indices among working age ($64$ years and below) and retiring age ($65$ years and above).}
\label{fig:gini-race-age}
\end{figure}


\section*{Concluding remarks}

In this paper, we propose a new multivariate extension of the Lorenz curve. We propose to emulate the \cite{gastwirth:1971} formulation of the Lorenz curve and define a Lorenz map by integrating vector quantiles of \cite{CGHH:2017}. The value of the Lorenz map is a vector of shares of each resource held by the poorer section of the population, as in the scalar case. Dominance of Lorenz maps defines a multi-attribute inequality dominance partial ordering. This Lorenz ordering is, like its scalar counterpart, an implementable criterion to compare inequality in allocations. It is, also like its scalar counterpart, equivalent to preference by any inequality averse rank dependent social evaluation functional. We propose an Inverse Lorenz Function and its level sets as a multivariate inequality visual comparison tool, and apply it to income and wealth in the United States between~$1989$ and~$2022$.

Multi-attribute inequality can vary substantially across population groups, as shown in \cite{MR:2016} within the information theoretic framework of \cite{Maasoumi:86}. 
There is a tension between heterogeneity across covariates and the anonymity axiom, according to which inequality measurement should not depend on individual's identities, but only on the distribution of resource allocations. As \cite{kolm:1977} pointed out, this tension is alleviated in part by including more variables in the allocation. This reinforces the motivation for a multidimensional approach to inequality measurement. 
As for the other potential sources of individual heterogeneity that matters to the social planner, anonymity can be restored by measuring inequality in each subgroup. We illustrate this in figure~\ref{fig:gini-race-age}. Beyond this, the conditional approach of \cite{MR:2016} could also be extended to our framework with the use of conditional vector quantiles in \cite{CCG:2016}.

Finally, we argue that a formal test of multi-attribute inequality dominance can be based on our Lorenz map, in analogy to dominance testing based on the traditional Lorenz curve in \cite{DD:2000} and references within. The statistical theory for such a test relies on multivariate stochastic dominance testing and the regularity of optimal transport maps, and is left for future research.

%


\appendix

{\small
\section{User's implementation guide}\label{sec:algo}

In this section, we will point to specific computational routines the reader may use to accomplish each step in section~\ref{sec:comp}. All of the figures in this paper were generated via implementations in the \texttt{R} language, however these implementations are standard and can be found in other languages and packages. Algorithm~\ref{algo:full} is therefore intended to guide the reader across the key steps of the implementation. We specialize to the case of $d=2$ and leave $d>2$ for the online supplement.

For the vector quantile, the \texttt{transport} package for \texttt{R} provides various implementations to solve \eqref{eq:Newton} that are found in the function \texttt{semidiscrete}. Both the standard descent approach of \cite{AHA:1998} and, our preferred, multiscale initialization and L-BGFS approach of \cite{merigot:2011} are supported methods. Alternatively and for all of our calculations, we used the \texttt{Rgeogram} package that is a wrapper of the C++ Geogram library implementation of \cite{merigot:2011}. Both packages provide the optimal weight vector $h$ required for the next steps.

The Lorenz map requires solving for the convex cells defined in \eqref{eq:cxcells}. The optimal~$h$ from the vector quantile calculation can be used as an input for the \texttt{power\_diagram} function in the \texttt{transport} package. It will provide as output the vertices of the convex cells. With a desired rank $r = (r_1,r_2)$, the next step is to find the area of the intersection of the convex cells with the rectangle $[0,r]$. The \texttt{sf} package provides these tools: first, the function \texttt{st\_polygon} transforms the vertices of the cells and rectangle into a ``polygon'' object that can be then read into \texttt{st\_intersection} and \texttt{st\_area}, which calculate the intersection and the area, respectively. These two functions can be used to calculate $\lambda(W_i^h \cap [0,r])$ in \eqref{eq:EmpLorenzMap} for all $i = 1, \dots, n$.

The Inverse Lorenz Function is calculated as an empirical distribution function. To facilitate drawing $\alpha$-Lorenz curves, it is recommended to form a uniform grid of values and calculate the ecdf at each value, e.g., $\{0.01,\dots,0.99,1\}^2$. Each pair should correspond to a row and column in a matrix of ecdf values. Then, pass this matrix as input into any function that plots contours of three-dimensional surfaces such as the base \texttt{R} function \texttt{contour} or \texttt{geom\_contour} as part of the \texttt{ggplot2} package. Finally, with a sample of Lorenz map values, one can compute the Gini index \eqref{eq:Gini} by plug-in.

\begin{algorithm}
\caption{Vector Quantiles, Lorenz Maps, and Inverse Lorenz Functions}\label{algo:full}
\begin{algorithmic}[1]
\Input
	\Desc{$(X,w)$}{\hspace{5mm} Weighted sample with normalized points $X_i$ and weights $w_i$}\newline
\EndInput

\Procedure{Vector-Quantile}{${X},{w}$}
\State Set convergence tolerance $\delta$, step size $\eta$.
\State  Set $s \gets 0$, initialize weight vector $h^0$. \Comment{e.g. multiscale approach}
\Repeat\Comment{Begin gradient descent}
    \State $s \gets s + 1$
    \State $h^{(s)} \gets h^{(s-1)} - \eta [w_i - \lambda(W_i^{h^{(s-1)}})]$ \Comment{Modifiable, e.g., L-BFGS}
\Until{$\norm{h^{(s)} - h^{(s-1)}} < \delta$}
\State $h := h^{(s)}$ solution to \eqref{eq:Newton} with $(X_i,w_i)$
\State $\{W_i^h\}_{i=1}^n \gets$ \eqref{eq:constraint} \Comment{e.g., computing power diagram using $h$}
\State  \textbf{return} $\{W_i^h\}$ defined by their vertices   \newline
\EndProcedure

\Procedure{Lorenz-Map}{$r,{X}, \{W_i^h\}_{i=1}^n$}
\Input
\Desc{$r$}{Vector of ranks of interest from $[0,1]^2$}
\Desc{$W_i^h$}{Vertices of cells that define the vector quantile of $(X,w)$}
\EndInput
\Output
\Desc{$\mathcal{L}_X$}{Lorenz map evaluated at $r$, a vector}
\EndOutput
\State $\mathcal{L}_X(r) \gets 0$
\For{$i = 1$ to $n$}
    \State Find vertices of $A_i := W_i^h \cap [0,r_1]\times[0,r_2]$ 
    \Comment{$A_i$ is convex}
    \State $\lambda_i:= \lambda(A_i)$: ordinary area of $A_i$ \Comment{Many equate to 0 or $w_i$}
    \State $\mathcal{L}_X(r) \gets \mathcal{L}_X(r) + X_i\lambda_i$
\EndFor
\State \textbf{return} $\mathcal{L}_X(r)$\newline
\EndProcedure

\Procedure{Inverse-Lorenz-Function}{${X}, \{W_i^h\}_{i=1}^n, m$}
\Input
    \Desc{$W_i^h$}{Vertices of cells that define the vector quantile of $(X,w)$}
    \Desc{$m$}{Size of pseudo sample from $U([0,1]^2)$}
\EndInput
\Output
\Desc{$l_X$}{Matrix of cumulative probabilities: rows and columns are coordinates}
\EndOutput
    \State Generate $Z:=$ evenly-spaced lattice in $[0,1]^2$ \Comment{e.g., $Z=\{0.01,\dots, 0.99,1\}^2$}
    \For{$j = 1$ to $m$}
        \State Draw a single $R_j \sim U[0,1]^2$
        \State $\mathcal{L}_j:=$ \textsc{Lorenz-Map}($R_j,X,\{W_i^h\}_{i=1}^n$)
    \EndFor
    \For{$z_{ij} \in Z$}
        \State ${\textstyle (l_X)_{ij} := m^{-1}\sum_{j=1}^m \mathds{1}\{\mathcal{L}_j \leq z_{ij}\}}$\Comment{empirical cdf of the $\mathcal{L}_j$}
    \EndFor
    \State\textbf{return} $l_X$ \Comment{Matrix of function values is usual input for contour plots}
\EndProcedure

\end{algorithmic}
\end{algorithm}


\section{Specific features and issues with the data source}
\label{app:data}

We review some known issues with the data set that impact our analysis. See \cite{HKL:2018} for a more in-depth account.

\subsection*{Sampling strategy}

The over sampling of high income and wealthy households is achieved by applying two distinct sampling techniques. The first sample is the core representative sample selected by a standard multi-stage area-probability design. The second is the high income supplement from statistical records derived from tax data by the Statistics of Income (SOI) division of the U.S. Internal Revenue Service. The stages sample disproportionately-- usually one-third of the final sample is from the high income supplement. Sampling in this way retains characteristic information of the population while also addressing the known selection biases of the wealthy not responding to surveys. In order to represent the population with this sample, weights must be constructed for each unit of observation. For more details on the construction of weights and their implications on the distribution of wealth, see \cite{KW:1999}.

\subsection*{Unit of observation and timing of interviews} 

The observations in this data set are not households, but rather a subset called the \emph{primary economic unit} (PEU) that may be individuals or couples and their financial dependents. For example in the 2016 data set 13\% of PEUs were in a household that contained one or more members not in their PEU. Additionally, the respondent is not necessarily the head of the household, so special care must be taken if analyzing attitudes in relation to some demographic characteristics such as age. 
The interviews start in May of the survey year, after most income taxes are filed and usually finish by the end of the calendar year, see \cite{kennickell:2017b} for challenges at the end of the interview period. Questions also may change over time so it is important to review the codebook each year when making comparisons across time. 

\subsection*{Multiple Imputation} 

During interviews, respondents may omit answers or provide a range of values for which their response belongs. This missing data impacts analysis and so the SCF contains 5 imputed values for each PEU, creating a sample 5 times larger than the actual number of respondents and forms 5 data sets called implicates. Imputation is done by the Federal Reserve Imputation Technique Zeta model (FRITZ), details can be found in \cite{kennickell:2017} based upon the ideas of \cite{rubin:2019}. 
Multiple imputation for missing data provide multiple probable values. Each of these form a data set from which sample statistics can be found. The technique of Repeated Imputation Inference (RII) is applied in our analysis. For each implicate $\ell = 1,2,3,4,5$, the empirical Inverse Lorenz Function $\widehat{l}_{\ell *}$ is calculated using the appropriate quantile map estimator taking into account sample weights. Then the repeated-imputation estimate of $l$ is
\[
	\widehat{l}(z) = \frac{1}{5}\sum_{\ell = 1}^5 \widehat{l}_{\ell *}(z).
\]
Calculation of the Gini index follows a similar procedure. Accounting for the multiple imputation in the calculation of standard errors is an important issue, but is not revelant to our visualization technique. For more information on multiple imputation and inference with imputed values, see \cite{rubin:1996}.

\subsection*{Definition of Wealth} 

In the literature, there is no consensus on what factors should be included in wealth measurement. \cite{wolff:2021} defines wealth as marketable weath, which is the sum of marketable or fungible assets less the current value of all debts. \cite{bricker:2017} define wealth as net worth including those assets which are not readily transformed into consumption: properties, vehicles, etc. In our analysis we consider all assets, including financial, as our wealth variable. 


\section{Additional details and results}


\subsection{Vector ranks and quantiles}
\label{app:VQ}

Proposition~\ref{prop:polar} below, a seminal result in the theory of
measure transportation (see \citeauthor{villani:2003} (\citeyear{villani:2003}, \citeyear{villani:2009})), states essential uniqueness of the gradient of a
convex function (hence cyclically monotone map) that pushes the uniform distribution on~$[0,1]^d$ into the distribution of an allocation~$X$.

Following %
\citet{villani:2003}, we let~$g_\#\nu $ denote the \emph{image measure} (or 
\emph{push-forward}) of a measure~$\nu$ by a measurable map~$g:\mathbb{R}%
^{d}\rightarrow \mathbb{R}^{d}$. Explicitly, for any Borel set~$A$, $g_\#\nu
(A):=\nu (g^{-1}(A))$. 
The symbol~$\nabla$ denotes the gradient, and~$D$ the Jacobian.
The convex conjugate of a convex lower semicontinuous function~$\psi$ is
denoted~$\psi^\ast$. 

\begin{proposition}[\citealt{mccann:1995}]
\label{prop:polar} Let $P$ and $\lambda $ be two distributions on $\mathbb{R}%
^{d} $. ($i$) If $\lambda $ is absolutely continuous with respect to the Lebesgue
measure on~$\mathbb{R}^{d}$, with support contained in a convex set $%
\mathcal{U}$, the following statements hold: there exists a convex function $\psi :%
\mathcal{U}\rightarrow \mathbb{R}\cup \{+\infty \}$ such that $\nabla \psi_\#\lambda =P$. The function $\nabla \psi $ exists and is unique, $\lambda $-almost
everywhere. ($ii$) If, in addition, $P$ is absolutely continuous on $\mathbb{R}%
^{d}$ with support contained in a convex set $\mathcal{X}$, the following
holds: there exists a convex function~$\psi ^{\ast }:\mathcal{X}\rightarrow 
\mathbb{R}\cup \{+\infty \}$ such that $\nabla \psi^{\ast }_\#P=\lambda $. The
function $\nabla \psi ^{\ast }$ exists, is unique and equal to $\left(
\nabla \psi \right) ^{-1}$, $P$-almost everywhere.
\end{proposition}

Proposition~\ref{prop:polar} is an extension of \cite{Brenier:91} (see also 
\cite{RR:90}). It removes the finite variance requirement, which is
undesirable in our context. Proposition~\ref{prop:polar} is the basis for
the definition of vector quantiles and ranks in \cite{CGHH:2017}. In our
context, it is applied with uniform reference measure.\footnote{%
This vector quantile notion was introduced in \cite{GH:2012} and \cite%
{EGH:2012} and called $\mu$-quantile.}

In case~$d=1$, gradients of convex functions are nondecreasing functions,
hence vector quantiles and ranks reduce to
classical quantile and cumulative distribution functions. As the notation
indicates, the function~$\psi ^{\ast }$ of proposition~\ref{prop:polar} 
is the convex conjugate of~$\psi $. In case of absolutely continuous distributions~$P$ on~$%
\mathbb{R}^{d}$ with finite variance, the vector rank function solves a
quadratic optimal transport problem, i.e., vector rank~$R$ minimizes, among
all functions~$T$ such that~$T(X)$ is uniform on~$[0,1]^{d}$, the quantity~$%
\mathbb{E}\Vert X-T(X)\Vert ^{2}$, where~$X\sim P$.

Proposition~\ref{prop:polar} is the basis for definition~\ref{def:VQ}. In the proofs, we shall use the notation~$Q_X=\nabla\psi_X$ for the vector quantile of a random vector~$X$ and call convex function~$\psi_X$ 
the transport potential associated with the distribution of~$X$.


\subsection{Egalitarian multi-attribute allocations}
\label{sec:egal}

\subsubsection{Identical allocations: additional details}
In this section, we consider bivariate allocations only. A sufficient condition for assumption~\ref{+regdep} is supermodularity of
the potential function~$\psi_X$ of allocation~$X$, as shown in lemma~\ref{lemma:PQD} below. We also show in lemma~\ref{lemma:PQD}, that
supermodularity of the potential function~$\psi_X$ also implies positive quadrant
dependence of the two components~$X_{1}$ and~$X_{2}$ of~$X$, i.e.,~$\mathbb P(X_1\leq x_1,X_2\leq x_2)\geq \mathbb P(X_1\leq x_1)\mathbb P(X_2\leq x_2)$, for all~$x_1,x_2\in\mathcal X$, see \cite{lehmann:1966}.

\begin{lemma}[Supermodular potential]
\label{lemma:PQD}
Suppose $X$ has a supermodular potential function, i.e.,
\begin{eqnarray*}
    \mathbb P(\partial^2\psi_X(U)/\partial u_1\partial u_2\geq0)=1,
\mbox{ with } U\sim U[0,1]^2.
\end{eqnarray*}
Then, assumption~\ref{+regdep} holds, and~$X_{1}$ and $X_{2}$ are positive quadrant dependent.
\end{lemma}

For allocations satisfying assumption~\ref{+regdep}, we show that Lorenz map and Inverse Lorenz Function of the identical allocation serve as upper and lower bounds, respectively.
Without assumption \ref{+regdep}, some allocations may have a Lorenz map that is
component-wise larger than the Lorenz map of the identical allocation for some ranks.
To illustrate the point, consider the potential~$\psi_X(u) = (u_1-u_2)^2/2 +u_1+u_2$. It corresponds to an allocation~$X$, whose distribution is
supported on the line~$X_1+X_2=2$. Calculating the Lorenz map, we obtain
\begin{eqnarray*}
\Lmap(r) & = & 
\left[
\begin{array}{c}
r_1r_2(r_1-r_2)/2+r_1r_2 \\ \\
r_1r_2(r_2-r_1)/2+r_1r_2
\end{array}
\right].
\end{eqnarray*}
Notice, in particular, that~$\Lmap_1(r) >r_1r_2$ in the region where $r_1>r_2$.
If the implicit relative price of resource~$2$ is~$1$, allocation~$X$ is an egalitarian allocation, since all individuals have equal budgets. However, this allocation does not satisfy assumption~\ref{+regdep} and its Lorenz map is not dominated by~$(r_1r_2,r_1r_2)$ as we have shown. This apparent departure from properties of the scalar Lorenz curve is due to the fact that an allocation with~$X_1+X_2=2$ a.s. can also be considered egalitarian, as we discuss in the following section. 

\subsubsection{Egalitarian allocations}

The identical allocation with Lorenz map~$(r_1r_2,r_1r_2)$ is a very special instance of egalitarian allocation. We extend this narrow notion of egalitarian allocation to include income egalitarianism, in the terminology of \cite{kolm:1977}.
In the special case where the two resources are transferable with relative price~$p$ of the second resource, an allocation is deemed egalitarian if all agents have the same budget endowment, i.e., if~$X_1+pX_2=1+p$ (where the constant value~$1+p$ is derived from the normalization~$\mathbb EX_1=\mathbb EX_2=1$). In the general case of non (or imperfectly) transferable resources, we call egalitarian the allocations with equalized shadow budgets.

\begin{definition}[Egalitarian allocation]
\label{def:egal}
An allocation~$X$ such that~$X_1+pX_2=1+p$ a.s., for some~$p>0$, is called egalitarian.
\end{definition}

Another way to interpret egalitarianism of such an allocation, beyond shadow budget equality, is through the perfect compensation of inequality in the marginal resource allocations by perfect negative correlation between resource allocations. The vector quantile and Lorenz map of egalitarian allocations can be characterized in the following way.

\begin{proposition}
\label{prop:egal}
Let~$(U_1,U_2)$ be a random vector with distribution~$U[0,1]^2$. $(i)$ An egalitarian allocation~$X$ such that~$X_1+pX_2=1+p$, admits potential~$\psi_X(u_1,u_2)=u_1+u_2+v(pu_1-u_2)$ for some convex function~$v$ such that~$\int_0^1v(p-z)dz=\int_0^1v(z)dz$ and allocation~$X$ is equal in distribution to~$(1+pv^\prime(pU_1-U_2),1 - v^\prime(pU_1-U_2))$; $(ii)$ The Lorenz map is given by
\begin{eqnarray*}
\mathcal L_X(r) & = & 
\begin{bmatrix}
r_1r_2 + \int_0^{r_2} [v(pr_1-u_2)-v(-u_2)]\;du_2 \\
r_1r_2 -\frac{1}{p} \int_0^{r_2} [v(pr_1-u_2)-v(-u_2)]\;du_2
\end{bmatrix};
\end{eqnarray*}
$(iii)$ If, in addition, $F_1^{-1}$ denotes the quantile function of~$X_1$, then
\begin{eqnarray*}
v(z) & = & \frac{1}{p} \int_0^z \left( F_1^{-1}(H_p(y))-1 \right) dy,
\end{eqnarray*}
where~$H_p$ is the cdf of the random variable~$pU_1-U_2$; see lemma~\ref{H_p} below for an explicit expression for~$H_p(z)$.
\end{proposition}

\begin{lemma}[Explicit formula for~$H_p(z)$]
\label{H_p}
The cumulative distribution function of~$Z=pU_1-U_2$ with~$(U_1,U_2)\sim U[0,1]^2$ is given by the following.
\begin{eqnarray*}
H_p(z) & = & \left\{ \begin{array}{llc} 
1 & \mbox{if} & p<z, \\ \\
1-\frac{p}{2}+z-\frac{z^2}{2p} & \mbox{if} & \max\{p-1,0\}<z\leq p, \\ \\
 \frac{1+2z}{2p} & \mbox{if} & 0<z\leq \max\{p-1,0\}, \\ \\
1 - \frac{p}{2} + z & \mbox{if} & \min\{p-1,0\}<z\leq 0, \\ \\
\frac{1}{p}\left(\frac{1}{2}+z+\frac{z^2}{2}\right) & \mbox{if} & -1< z\leq \min\{p-1,0\},\\ \\
0 & \mbox{if} & z\leq -1.
\end{array}
\right.
\end{eqnarray*}
\end{lemma}

We see in proposition~\ref{prop:egal} that the distribution of the egalitarian allocation~$X$ is entirely determined by the convex function~$v$, which is itself determined by the distribution of one of the marginals of~$X$. This follows from the deterministic linear relationship between the two resource allocations. The perfect negative correlation compensates any inequality in the marginal allocations. 

With this definition of egalitarian allocations, we show that a large class of allocations are dominated in the Lorenz order by egalitarian allocations, and that egalitarian allocations are maximal in the Lorenz order of definition \ref{def:s-order}.

\begin{assumption}
\label{ass:egal}
For some~$p>0$, the potential~$\psi_X$ of allocation~$X$ satisfies for all~$z\in[-1,p]$:
\begin{eqnarray*}
\begin{array}{lcl}
\sup_{pu_1-u_2=z}\left\{ -\frac{1}{p}\frac{\partial^2\psi_X(u_1,u_2)}{\partial u_1\partial u_2}\right\} & \leq & \inf_{pu_1-u_2=z}\min\left\{\frac{1}{p^2}
\frac{\partial^2\psi_X(u_1,u_2)}{\partial u_1^2},\frac{\partial^2\psi_X(u_1,u_2)}{\partial u_2^2}\right\}.
\end{array}
\end{eqnarray*} 
\end{assumption}

Before stating the main result of this section, which is an extension of property~(7) in section~\ref{sec:props}, we discuss sufficient conditions for assumption~\ref{ass:egal} and examples of classes of allocations that satisfy assumption~\ref{ass:egal}. The following lemma provides sets of sufficient conditions based on a suitable choice of~$p$.
\begin{lemma}[Sufficient condition for assumption~\ref{ass:egal}]
\label{lemma:egal-suff}
An allocation with potential~$\psi_X$ satisfies assumption~\ref{ass:egal} if any of the following conditions hold.
\begin{enumerate}
\item[$(i)$] The potential~$\psi_X$ is supermodular.
\item[$(ii)$] The potential~$\psi_X$ satisfies:
\begin{eqnarray}
\label{eq:quad}
 \sqrt{ \inf_{u_1,u_2} \frac{\partial^2\psi_X(u_1,u_2)}{\partial u_1^2} \times \inf_{u_1,u_2}\frac{\partial^2\psi_X(u_1,u_2)}{\partial u_2^2} }
 + \inf_{u_1,u_2} \frac{\partial^2\psi_X(u_1,u_2)}{\partial u_1\partial u_2} \geq 0.
\end{eqnarray}
\item[$(iii)$] The function
\begin{eqnarray*}
p(u_1,u_2) & := & \sqrt{ \frac{\partial^2\psi_X(u_1,u_2)}{\partial u_1^2} \bigg/ \frac{\partial^2\psi_X(u_1,u_2)}{\partial u_2^2} }
\end{eqnarray*}
is positive and constant equal to~$p$ over~$[0,1]^2$ and, for all~$z\in[-1,p]$, the Hessian of~$\psi_X$ is constant over~$pu_1-u_2=z$.
\end{enumerate}
\end{lemma}

The first sufficient condition in lemma~\ref{lemma:egal-suff}, i.e., supermodularity of the potential~$\psi_X$, imposes a form of positive dependence between the two resources, which implies assumption~\ref{ass:egal} (and \ref{+regdep}). However, assumption~\ref{ass:egal} also accommodates allocations that do not exhibit positive dependence. For instance, the mixture of an egalitarian allocation with a positively dependent one satisfies assumption~\ref{ass:egal}.
\begin{example}
\label{ex:egal2}
An allocation with potential~$\psi(u_1,u_2)=u_1+u_2+v(pu_1-u_2)+\tilde\psi(u_1,u_2)$, with~$v$ convex,~$p>0$ and~$\tilde{\psi}$ ultramodular, satisfies assumption~\ref{ass:egal}. It mixes a perfectly negatively correlated allocation with a positively dependent one.
\end{example}
Aspecial case of condition~(2) in lemma~\ref{lemma:egal-suff} is the case, where~$\psi_X$ is a quadratic function, hence has a constant Hessian. Indeed, in that case, convexity of~$\psi_X$ immediately yields~(\ref{eq:quad}). 
\begin{example}
All allocations with quadratic potential~$\psi_X(u_1,u_2)=a_1u_1+a_2u_2+a_{11}u_1^2+a_{12}u_1u_2+a_{22}u_2^2$ with $a_{1},a_{11}, a_2, a_{22}, a_{12} \in \mathbb{R}$, i.e., allocations of the form~$X=(a_1+2a_{11}U_1+a_{12}U_2,a_2+2a_{22}U_2+a_{12}U_1)$, with~$(U_1,U_2)\sim U[0,1]^2$, satisfy assumption~\ref{ass:egal}.
\end{example}
Sufficient condition~(2) in lemma~\ref{lemma:egal-suff}  can also be used to show that allocations where the two marginal resource allocations are independent also satisfy assumption~\ref{ass:egal}. More generally, a large class of allocations defined as deviations from independence satisfy assumption~\ref{ass:egal} as formalized in the following example.
\begin{example}
An allocation~$X$ with potential~$\psi_X(u_1,u_2)=\phi_1(u_1)+\phi_2(u_2) +\rho \phi(u_1,u_2)$ satisfies assumption~\ref{ass:egal} if~$\phi_1'' \geq B_1, \phi_2'' \geq B_2, \frac{\partial^2\phi }{\partial u_1^2} \geq B_{11}, \frac{\partial^2\phi }{\partial u_2^2} \geq B_{22}$, $\frac{\partial^2\phi }{\partial u_1u_2} \geq B_{12}$, and~$-\rho B_{12} \leq \sqrt{(B_1+\rho B_{11})(B_2+\rho B_{22}) } $ with $B_{1},B_{11}, B_2, B_{22}, B_{12} \in  \mathbb{R}$. The case~$\rho=0$ is the case of independent marginal allocations.
\end{example}

Assumption~\ref{ass:egal} is not satisfied, however, in case~$X_1$ and~$X_2$ are perfectly negatively dependent, i.e., $X_2=-\phi(X_1)$ with increasing~$\phi$, when~$\phi$ is nonlinear.

Under assumption~\ref{ass:egal}, we can complement property~(7) in section~\ref{sec:props} and emulate the traditional property of Lorenz curves, which are maximal at perfect equality. Here we show that egalitarian allocations dominate all allocations that satisfy assumption~\ref{ass:egal}, and are themselves undominated thereby forming a class of distributions that are maximal under the Lorenz order.

\begin{continued}[Property~(7) continued][Lorenz map maximal at egalitarian allocations]
\label{P3}
For any allocation~$X$ satisfying assumption~\ref{ass:egal}, there is an egalitarian allocation~$\tilde X$ such that~$X\preccurlyeq_{\mathcal L}\tilde X$, i.e.,~$\mathcal L_X(r)\leq\mathcal L_{\tilde X}(r)$ for all~$r\in[0,1]^2$.
In addition, if two egalitarian allocations are ranked in the Lorenz order, then they are equal. \end{continued}


\subsubsection{Proofs for section \ref{sec:egal}}


\begin{proof}[Proof of lemma~\ref{lemma:PQD}]
Let~$U\sim U[0,1]^2$. Let~$\tilde X_j:=\frac{\partial\psi}{\partial u_j}(U_1,U_2)$, $j=1,2$. Then~$(\tilde X_1,\tilde X_2)$ is distributed identically to~$(X_1,X_2)$. Since $\tilde{X}_{2}=\frac{\partial \psi }{\partial u_{2}}(U_{1},U_{2})$
is monotonically increasing in $U_{2}$, we have $U_{2}=\left( \frac{\partial
\psi }{\partial u_{2}}\right) ^{-1}(U_{1},\tilde{X}_{2})$. Hence 
\begin{eqnarray*}
\tilde{X}_{1} & = & \frac{\partial \psi }{\partial u_{1}}\left( U_{1},\left( 
\frac{\partial \psi }{\partial u_{2}}\right) ^{-1}(U_{1},\tilde{X}%
_{2})\right) .
\end{eqnarray*}%
Under the stated assumption, $\tilde{X}_{1}$ is increasing in $U_{1}$
and $\tilde{X}_{2}$. Since $\left( \tilde{X}_{1},\tilde{X}%
_{2}\right) \overset{d}{=}\left( X_{1},X_{2}\right) $, we have 
\begin{eqnarray*}
F_X\left( x_{1},x_{2}\right) &=&\mathbb{P} \left( \tilde{X}_{1}\leq x_{1},%
\tilde{X}_{2}\leq x_{2}\right) \\
&=&\mathbb{E}\bigg[ \mathbb{P} \left( \frac{\partial \psi }{\partial u_{1}}\left(
U_{1},\left( \frac{\partial \psi }{\partial u_{2}}\right) ^{-1}(U_{1},%
\tilde{X}_{2})\right) \leq x_{1},\tilde{X}_{2}\leq x_{2}\right)\bigg\vert \; U_{1}\bigg] \\
&=&\mathbb{E}\left[ \min \left\{ F_{1}\left( x_{1}|U_{1}\right) ,F_{2}\left(
x_{2}|U_{1}\right) \right\} \right] \\
&\geq &\mathbb{E}\left[ F_{1}\left( x_{1}|U_{1}\right) F_{2}\left( x_{2}|U_{1}\right) %
\right] ,
\end{eqnarray*}%
where~$F_i(\cdot\vert U_j)$ denotes the cumulative distribution function of~$X_i$ conditional on~$U_j$.
Now $F_{1}\left( x_{1}|U_{1}\right) $ is increasing in $U_{1}$, since 
\begin{eqnarray*}
F_{1}\left( x_{1}|U_{1}\right) &=&\mathbb{P} \left( \frac{\partial \psi }{\partial
u_{1}}(U_{1},U_{2})\leq x_{1}\bigg\vert \; U_{1}\right) \\
&=&\mathbb{P} \left( U_{2}\leq \left( \frac{\partial \psi }{\partial u_{1}}\right)
^{-1}(U_{1},x_{1})\bigg\vert \; U_{1}\right) \\
&=&\left( \frac{\partial \psi }{\partial u_{1}}\right) ^{-1}(U_{1},x_{1}).
\end{eqnarray*}%
Similarly $F_{2}\left( x_{2}|U_{1}\right) $ is increasing in $U_{1}$. We
conclude that $F_X\left( x_{1},x_{2}\right) \geq F_{1}\left( x_{1}\right)
F_{2}\left( x_{2}\right) $, see e.g. \cite{joe:97}. 
\end{proof}

\begin{proof}[Proof of proposition~\ref{prop:egal}]

The potential~$\psi$ of an egalitarian allocation satisfies~$\partial\psi/\partial u_1+p\partial\psi/\partial u_2=1+p$. Solutions are of the form
\begin{eqnarray*}
\psi_{(v,p)}(u_1,u_2)=u_1+u_2+v(pu_1-u_2).
\end{eqnarray*}
Convexity of~$\psi$ implies convexity of~$v$. The normalization
\begin{eqnarray*}
\int_0^1\!\!\!\int_0^1\nabla\psi_{(v,p)}(u_1,u_2)du_1du_2=(1,1)
\end{eqnarray*} 
implies
\begin{eqnarray*}
\int_0^1\!\!\!\int_0^1v^\prime(pu_1-u_2)du_1du_2=0.
\end{eqnarray*} 
The latter, in turn, implies
\begin{eqnarray*}
\int_0^1v(p-x)dx=\int_0^1v(x)dx.
\end{eqnarray*}
Call~$H_p$ the cdf of~$Z=pU_1-U_2$, where~$(U_1,U_2)\sim U[0,1]^2$. Call $F_1$ the cdf of~$\nabla_1\psi_{(p,v)}:=1+pv^\prime(Z)$, which is the first marginal of allocation~$X$.
Then
\begin{eqnarray*}
F_1(x) & = & \mathbb P \left( v^\prime(Z) \leq \frac{x-1}{p} \right) \\
&=& \mathbb P \left( Z \leq (v^\prime)^{-1}\left( \frac{x-1}{p} \right) \right) \\
& = & H_p \left( (v^\prime)^{-1}\left( \frac{x-1}{p} \right) \right).
\end{eqnarray*}
Now
\begin{eqnarray*}
F_1(x)  =  H_p \left( (v^\prime)^{-1}\left( \frac{x-1}{p} \right) \right) & \Rightarrow & (v^\prime)^{-1}\left( \frac{x-1}{p} \right) = H_p^{-1}\left(F_1(x)\right) \\
& \Rightarrow & \frac{x-1}{p} = v^\prime\left( H_p^{-1}\left(F_1(x)\right) \right) \\
& \Rightarrow & v^\prime(z) = \frac{F_1^{-1}(H_p(z))-1}{p}.
\end{eqnarray*}
Hence
\begin{eqnarray*}
v(z) & = & \int_0^x  \frac{F_1^{-1}(H_p(y))-1}{p}dy,
\end{eqnarray*}
as desired.
\end{proof}

\begin{proof}[Proof of lemma~\ref{lemma:egal-suff}]
A sufficient condition for assumption~\ref{ass:egal} is
\begin{eqnarray*}
 -\inf_{u_1,u_2} \frac{1}{p}\frac{\partial^2\psi_X(u_1,u_2)}{\partial u_1\partial u_2} \leq \min \left\{ \inf_{u_1,u_2} \frac{1}{p^2}\frac{\partial^2\psi_X(u_1,u_2)}{\partial u_1^2},\inf_{u_1,u_2}\frac{\partial^2\psi_X(u_1,u_2)}{\partial u_2^2}\right\}.
 \end{eqnarray*}
If we choose the optimal value of~$p$, i.e., 
\begin{eqnarray*}
p^2=\frac{ \inf_{u_1,u_2} \frac{\partial^2\psi_X(u_1,u_2)}{\partial u_1^2}}{\inf_{u_1,u_2}\frac{\partial^2\psi_X(u_1,u_2)}{\partial u_2^2} },
\end{eqnarray*}
we get the sufficient inequality
\begin{eqnarray*}
-\inf_{u_1,u_2} \frac{\partial^2\psi_X(u_1,u_2)}{\partial u_1\partial u_2} \leq \sqrt{ \inf_{u_1,u_2} \frac{\partial^2\psi_X(u_1,u_2)}{\partial u_1^2} \times \inf_{u_1,u_2}\frac{\partial^2\psi_X(u_1,u_2)}{\partial u_2^2} }
\end{eqnarray*}
as desired.
\end{proof}


\begin{proof}[Proof of example~\ref{ex:egal}]
Let~$\psi(u_1,u_2)=u_1+u_2+v(pu_1-u_2)+\tilde\psi(u_1,u_2)$, with~$v$ convex and twice continuously differentiable, and~$\psi$ ultramodular. We have, for~$j=1,2,$
\begin{eqnarray*}
\begin{array}{lll}
\frac{\partial^2\psi(u_1,u_2)}{\partial u_1^2} 
& = &
p^2v''(pu_1-u_2)+ 
\frac{\partial^2\tilde\psi(u_1,u_2)}{\partial u_2^2},\\ \\
\frac{\partial^2\psi(u_1,u_2)}{\partial u_1^2} 
& = &
v''(pu_1-u_2)+ 
\frac{\partial^2\tilde\psi(u_1,u_2)}{\partial u_2^2}.
\end{array}
\end{eqnarray*}
Also,
\begin{eqnarray*}
\begin{array}{lll}
\frac{\partial^2\psi(u_1,u_2)}{\partial u_1\partial u_2} 
& = &
-pv''(pu_1-u_2)+\frac{\partial^2\tilde\psi(u_1,u_2)}{\partial u_1\partial u_2}.
\end{array}
\end{eqnarray*}
Therefore
\begin{eqnarray*}
\begin{array}{rll}
\sup_{pu_1-u_2=z}\left\{-\frac{1}{p}\frac{\partial^2\psi(u_1,u_2)}{\partial u_1\partial u_2}\right\} 
& = &
\sup_{pu_1-u_2=z}\left\{v''(pu_1-u_2)-\frac{1}{p}\frac{\partial^2\tilde\psi(u_1,u_2)}{\partial u_1\partial u_2}\right\} \\ \\
& = & 
v''(z)-\inf_{pu_1-u_2=z}\left\{\frac{1}{p}\frac{\partial^2\tilde\psi(u_1,u_2)}{\partial u_1\partial u_2}\right\}  \\ \\
 \leq \;\;
v''(z) & + & \min\left\{ \inf_{pu_1-u_2=z}\frac{1}{p^2}\frac{\partial^2\tilde{\psi}(u_1,u_2)}{\partial u_1^2}, \inf_{pu_1-u_2=z}\frac{\partial^2\tilde{\psi}(u_1,u_2)}{\partial u_2^2} \right\} \\ \\
& = & 
\inf_{pu_1-u_2=z}\min\left\{ \frac{1}{p^2}\frac{\partial^2\psi(u_1,u_2)}{\partial u_1^2} ,\frac{\partial^2\psi(u_1,u_2)}{\partial u_2^2} \right\}.
\end{array}
\end{eqnarray*}

\end{proof}


\begin{proof}[Proof of ``property (7) continued'' in~\ref{P3}]
Define
\begin{eqnarray*}
v''(z) & := & \inf_{pu_1-u_2=z}\min\left\{
\frac{1}{p^2}\frac{\partial^2\psi_X(u_1,u_2)}{\partial u_1^2},\frac{\partial^2\psi_X(u_1,u_2)}{\partial u_2^2}\right\}.
\end{eqnarray*}
Under assumption~\ref{ass:egal},
\begin{eqnarray*}
v''(z) &\geq & \sup_{pu_1-u_2=z}\left\{-\frac{1}{p}
\frac{\partial^2\psi_X(u_1,u_2)}{\partial u_1\partial u_2}\right\}.
\end{eqnarray*}
Hence 
$\psi := \psi_X(u_1,u_2)-v(pu_1-u_2)$
is an ultramodular function. 
Applying the proof of property~(7) in section~\ref{sec:props}, we find that for all~$(r_1,r_2)\in[0,1]^2$,
\begin{eqnarray*}
\int_0^{r_1}\!\!\!\int_0^{r_2}\nabla\psi(u_1,u_2)du_1du_2\leq (r_1r_2,r_1r_2).
\end{eqnarray*}
Hence
\begin{eqnarray*}
\int_0^{r_1}\!\!\!\int_0^{r_2}\nabla\psi_X(u_1,u_2)du_1du_2\leq \int_0^{r_1}\!\!\!\int_0^{r_2}\nabla\psi_{(v,p)}(u_1,u_2)du_1du_2,
\end{eqnarray*}
where~$\psi_{(v,p)}(u_1,u_2):=v(pu_1-u_2)+u_1+u_2$ as desired.


We now show that egalitarian allocations do not dominate each other.
Suppose an egalitarian allocation~$X_{(v,p)}$ with potential~$v(pu_1-u_2)+u_1+u_2$ dominates an allocation~$X_{(\tilde v,\tilde p)}$ with potential~$\tilde v(\tilde pu_1-u_2)+u_1+u_2$.
Then
\begin{eqnarray*}
&&\left[
\begin{array}{cc} 
r_1r_2 +\int_0^{r_1}\!\!\! \int_0^{r_2} [1+p\tilde v'(\tilde pr_1-u_2)]\;du_1du_2 \\
r_1r_2 +\int_0^{r_1}\!\!\!\int_0^{r_2} [1-\tilde v'(\tilde pr_1-u_2)]\;du_1du_2
\end{array}
\right]
\\
&& \hskip100pt \leq  
\left[
\begin{array}{cc} 
r_1r_2 + \int_0^{r_1}\!\!\!\int_0^{r_2} [1+pv'(pr_1-u_2)]\;du_1du_2 \\
r_1r_2 +\int_0^{r_1}\!\!\!\int_0^{r_2} [1-v'(pr_1-u_2)]\;du_1du_2
\end{array}
\right].
\end{eqnarray*}
Hence, for all~$(r_1,r_2)\in[0,1]^2$,
\begin{eqnarray*}
\int_0^{r_1}\!\!\!\int_0^{r_2}\tilde v'(\tilde pu_1-u_2) & = & \int_0^{r_1}\!\!\!\int_0^{r_2}v'(pu_1-u_2),
\end{eqnarray*}
so that both allocations have the same Lorenz map, hence are equally distributed.
\end{proof}

\subsection{Uniform Majorization}
\label{app:UM}

In this section, we show the undesirable feature of uniform majorization detailed in section~\ref{sec:MRT}. \cite{GH:2012} show that the only rank dependent social evaluation functional (up to an affine transformation) that satisfies the uniform majorization principle of \cite{kolm:1977} is given in~(\ref{eq:UM-Gini}). We now show that it is unsuitable as a tool to measure multivariate inequality with the following two observations about bivariate allocations~$X$.
First, the following expression shows that~$S_{UM}$ only depends on ~$\Lmap_1(p,1)+\Lmap_2(1,p)$, hence, on very specific features of the dependence between the two components~$X_1$ and~$X_2$ of allocation~$X$.
\begin{eqnarray}
\label{eq:Gin1}
S_{UM}(X) & = & \int_{0}^{1} \Bigl( \left[ \Lmap_{1}\left( p,1\right) - p \right] + \left[ \Lmap_{2}\left( 1,p\right) -p \right] \Bigr) \,dp.
\end{eqnarray}

Second, and more troubling still, for any given fixed marginals for~$X_1$ and~$X_2$,~$S_{UM}$ is minimized when~$X_1$ and~$X_2$ are independent. Indeed, we show below that~$S_{UM}$ is always larger than minus the average of univariate Gini coefficients, which is its value when the marginals are independent.
\begin{eqnarray}
\label{eq:Gin2}
S_{UM}(X) & \geq & -\frac{1}{2}\left[ \, G\left( X_{1}\right)
+G\left( X_{2}\right) \, \right],
\end{eqnarray}
where~$G(X_1)$ and~$G(X_2)$ denote the classical scalar Gini index of components~$X_1$ and~$X_2$ respectively.

\begin{proof}[Proof of~(\ref{eq:Gin1})]
Let~$U$ be uniformly distributed on~$[0,1]^2$. Note that
\begin{eqnarray*}
S_{UM}(X) & = & 1-\mathbb E\left[ U_1\frac{\partial\psi}{\partial u_1}(U_1,U_2) \right]
- \mathbb E\left[ U_1\frac{\partial\psi}{\partial u_1}(U_1,U_2) \right].
\end{eqnarray*}
Now,
\begin{eqnarray*}
\mathbb{E}\left[ U_1 \frac{\partial \psi}{\partial u_1}(U_1,U_2) \right] 
& = & \int_0^1 \left[ \int_0^1 u_1 \frac{\partial \psi}{\partial u_1}(u_1,u_2)
\, du_1 \right] du_2 \\
& = & \int_0^1 \left[ \int_0^1 u_1 \,d\left(\frac{\partial \Lmap_1}{%
\partial u_2}(u_1,u_2)\right) \, \right] du_2,
\end{eqnarray*}
where the last equality follows from interchangeability of the order of
integration and~$\Lmap_1$ is the first
component of the Lorenz map.
Note that 
\begin{eqnarray*}
\frac{\partial \Lmap_1}{\partial r_2}(r_1,r_2) & = & \int_0^{r_1} \frac{\partial
\psi}{\partial u_1}(u_1,u_2) \, du_1
\end{eqnarray*}
and 
\begin{eqnarray*}
\frac{\partial}{\partial r_1}\left(\frac{\partial \Lmap_1(r_1,r_2)}{\partial r_2}%
\right) & = & \frac{\partial \psi}{\partial u_1}(r_1,r_2).
\end{eqnarray*}
Therefore 
\begin{eqnarray*}
\mathbb{E}\left[ U_1 \frac{\partial \psi}{\partial u_1}(U_1,U_2) \right] 
& = & \int_{0}^{1}\left( u_{1}\frac{\partial \Lmap_{1}}{\partial
u_{2}}(u_{1},u_{2}) \bigg\vert_{0}^{1}-\int_{0}^{1}\frac{\partial \Lmap_{1}}{%
\partial u_{2}}(u_{1},u_{2})\,du_{1} \right)\,du_{2} \\
& = & \int_{0}^{1}\left(\frac{\partial \Lmap_{1}}{\partial u_{2}}(1,u_{2})%
-\int_{0}^{1}\frac{\partial \Lmap_{1}}{\partial u_{2}}(u_{1},u_{2})\,du_{1}\right)%
\,du_{2} \\
& = & \int_{0}^{1}\frac{\partial \Lmap_{1}}{\partial u_{2}}(1,u_{2})%
\,du_{2}-\int_{0}^{1}\int_{0}^{1}\frac{\partial \Lmap_{1}}{\partial
u_{2}}(u_{1},u_{2})\,du_{2}\,du_{1} \\
& = & \Lmap(1,1)-\Lmap(1,0)-\int_{0}^{1}\Lmap_{1}(u_{1},1)-\Lmap_{1}(u_{1},0)\,du_{1} \\
& = & 1-\int_{0}^{1}\Lmap_{1}(u_{1},1)\,du_{1}.
\end{eqnarray*}
Similarly, we have $\mathbb{E}\left[ U_2 \frac{\partial \psi}{\partial u_2}(U_1,U_2) \right] = 1-\int_{0}^{1}\Lmap_{2}(1,u_{2})\,du_{1}$, as desired.
\end{proof}

\begin{proof}[Proof of~(\ref{eq:Gin2})]
The inequality follows from~$\Lmap_{1}\left( p,1\right) \geq L_{1}\left( p\right) $ and~$\Lmap_{2}\left( 1,p\right) \geq L_{2}\left( p\right)$. We now prove the latter.
Letting $\nabla \psi$ be the vector quantile function of $(X_1,X_2)$, note that since  $\frac{\partial \psi}{\partial u_1}$ pushes uniform measure on $[0,1]^2$ forward to law$(X_1)$, we can write 
$$
L_1(r_1) =\int_{\{u:\frac{\partial \psi}{\partial u_1}(u) \leq z_1(r_1)\}}\frac{\partial \psi}{\partial u_1}(u)du
$$
where $z_1(r_1)$ is the quantile of the random variable $X_1$.  Note that the area of the domain  $\{u:\frac{\partial \psi}{\partial u_1}(u) \leq z_1(r_1)\}$ of integration must be $r_1$. 
On the other hand, 
$$
\Lmap_1(r_1,1) =\int_0^{r_1}\int_0^1\frac{\partial \psi}{\partial u_1}(u)du_1du_2
$$
is an integral of the same function over a region with the same area.   Writing $A:=\left\{u:\frac{\partial \psi}{\partial u_1}(u) \leq z_1(r_1)\right\}$, we have $A=B\cup C$, where $B=A \cap ([0,r_1] \times [0,1])$ and $C=A \cap ((r_1,1] \times [0,1])$ and the union is disjoint. Similarly, $[0,r_1] \times [0,1] = B \cup D$ where $D =([0,r_1] \times [0,1]) \cap A^c$.   Note that the areas of $C$ and $D$ must be the same, $|C|=|D|$, and $\frac{\partial \psi}{\partial u_1}(u) \leq z_1(r_1)$ throughout $C$ while 
$\frac{\partial \psi}{\partial u_1}(u) > z_1(r_1)$ throughout $D$.  We have
\begin{eqnarray*}
L_1(r_1)&=&\int_{B}\frac{\partial \psi}{\partial u_1}(u)du +\int_{C}\frac{\partial \psi}{\partial u_1}(u)du\\
&\leq &\int_{B}\frac{\partial \psi}{\partial u_1}(u)du +z_1(r_1)|C|\\
&=&\int_{B}\frac{\partial \psi}{\partial u_1}(u)du +z_1(r_1)|D|\\
&\leq&\int_{B}\frac{\partial \psi}{\partial u_1}(u)du +\int_{D}\frac{\partial \psi}{\partial u_1}(u)du\\
&=&\Lmap_1(r_1,1).
\end{eqnarray*}
Note that this inequality holds for any dependence structure between $X_1$ and $X_2$.
\end{proof}

%


\section{Inequality Dominance based on the Inverse Lorenz Function}

We can also define an increasing inequality order based on the Inverse Lorenz Functions. Consider two allocations~$X$ and~$X^\prime$, with respective Inverse Lorenz Functions~$l_X$ and~$l_{X^\prime}$.
If~$l_X(z)\leq l_{X^\prime}(z)$ for some vector of shares~$z$, a larger proportion of the population commands the same share of resources in allocation~$X^\prime$ than in allocation~$X$. If this is true for any vector~$z$ of resource shares in~$[0,1]^d$, then, we say that allocation~$X^\prime$ is more unequal than allocation~$X$.

\begin{definition}
\label{def:order}
An allocation~$X^\prime$ is said to be more unequal in the weak Lorenz order than an allocation~$X$ if~$l_{X^\prime}(z )\geq l_X(z)$ for all~$z\in[0,1]^d$. We denote this~$X\succcurlyeq _lX^\prime$.
\end{definition}

The relation~$X\succcurlyeq_l X^\prime$ is equivalent to lower orthant dominance of the random vector~$\Lmap_X(U)$, with~$U\sim U[0,1]^d$, over~$\Lmap_{X^\prime}(U)$ (see Section~3.8 of \cite{MS:2002}). 
In the scalar case, the orderings of definitions~\ref{def:s-order} and~\ref{def:order} both coincide with the traditional Lorenz ordering. In higher dimensions, however, the equivalence may not hold\footnote{We have not been able to either prove the equivalence or find a counterexample.}. Nonetheless,
as the name indicates, the weak Lorenz inequality order of definition~\ref{def:order} is weaker than the Lorenz order of definition~\ref{def:s-order}, as we show in proposition~\protect\ref{prop:order}.

\begin{proposition}
\label{prop:order}
An allocation~$X^\prime$ is more unequal in the weak Lorenz order than an allocation~$X$, i.e., $X\succcurlyeq _lX^\prime$ (definition~\ref{def:order}) if~$X^\prime$ is more unequal in the Lorenz order, i.e., $X\preccurlyeq _{\Lmap}X^\prime$ (definition~\ref{def:s-order}).
\end{proposition}

\begin{proof}[Proof of proposition~\protect\ref{prop:order}]
$\tilde X\succcurlyeq _{\Lmap}X$ is equivalent to first order stochastic dominance of~$%
\Lmap_X(U)$ over~$\Lmap_{\tilde X}(U)$, where~$U\sim U[0,1]^d$ (see Section~6.B page~266
of \cite{SS:2007}). Hence,~$\tilde X\succcurlyeq _{\Lmap}X$ implies~$\mathbb{P}(\Lmap_X(U)\in S)\leq \mathbb{P}%
(\Lmap_{\tilde X}(U)\in S)$ for any lower set~$S$, so that~$\tilde X\succcurlyeq _{\Lmap}X$ implies~$\tilde X\succcurlyeq_lX$, given that the sets~$[0,z]$ are lower sets.
\end{proof}

\section{Proofs of the main results}

\label{sec:proofs}

Recall that in the proofs, we shall use the notation~$Q_X=\nabla\psi_X$ for the vector quantile of a random vector~$X$ and call convex function~$\psi_X$ 
the transport potential associated with the distribution of~$X$. See section~\ref{app:VQ} for details. In this section, we omit the~$X$ subscript of~$\psi_X$ for notational compactness.


\begin{proof}[Proof of proposition~\protect\ref{prop:char}]
In case~$d=2$, the off diagonal elements of the Jacobian of~$\Lmap(r)$ are~$\psi(r_1,r_2)-%
\psi(r_1,0)$ and~$\psi(r_1,r_2)-\psi(0,r_2)$. From the latter, by
differentiation, we can recover~$\nabla\psi(r_1,r_2)$. The result then follows from the fact that~$\nabla\psi$ characterizes~$P_X$, see for instance \cite{CGHH:2017}. The result extends straightforwardly to~$d>2$.
\end{proof}


\begin{proof}[Proof of proposition~\ref{prop:+regdep}]
We only need to show the result for one component of the Lorenz map and the others
follow with similar reasoning. We have for the first component 
\begin{eqnarray*}
\Lmap_1(r)&=& \int_0^{r_d}\cdots\int_0^{r_2}\int_0^{r_1} \frac{\partial \psi}{\partial u_1}%
(u_1,u_2,...,u_d)du_1du_2 ...du_d\\
&=&\int_0^{r_d}\cdots\int_0^{r_2} [\psi (r_1,u_2,...,u_d) - \psi (0,u_2,...,u_d)]du_2...du_d \\
&\leq & \int_0^{r_d}\cdots\int_0^{r_2} r_1[\psi (1,u_2,...,u_d) - \psi (0,u_2,...,u_d)]du_2...du_d \text{, by
convexity } \\
&=&r_1\int_0^{r_d}\cdots\int_0^{r_2}\int_0^{1} \frac{\partial \psi}{\partial u_1}(u_1,u_2,...,u_d)du_1
du_2...du_d
\end{eqnarray*}

Now define $H(r_2,...,r_d) = \int_0^{r_d}\cdots\int_0^{r_2}\int_0^{1} \frac{\partial \psi}{\partial u_1%
}(u_1,u_2,...,u_d)du_1 du_2...du_d$. Then $$
\frac{\partial H}{\partial r_2}(r_2,...,r_d)=\int_0^{r_d}\cdots\int_0^{r_3}\int_0^1%
\frac{\partial \psi}{\partial u_1}(u_1,r_2,u_3,...,u_d)du_1du_3...du_d
$$
is monotone increasing in $r_2$, since it is the integral of the functions $r_2 \mapsto \int_0^1\frac{\partial \psi}{\partial u_1}(u_1,r_2,u_3,...,u_d)du_1$, which are monotonically increasing by assumption. 
Therefore,  $%
r_2 \mapsto H(r_2,...r_d)$ is convex and so $H(r_2,r_3,...,r_d) \leq H(0,r_3,...,r_d) +r_2(H(1,r_3,...,r_d)-H(0,r_3,...,r_d))$. Note that $H(0,r_3,...,r_d)
= 0$ as an integral over a degenerate interval, so $H(r_2,r_3,...,r_d)\leq r_2H(1,r_3,...,r_d)$.  A very similar argument yields $H(1,r_3,...,r_d) \leq r_3H(1,1,r_4,...,r_d)$, so that $H(r_2,r_3,...,r_d)\leq r_2r_3H(1,1,r_4,...,r_d)$, and, iterating in this way, we eventually obtain,
$$
H(r_2,r_3,...,r_d)\leq r_2r_3\cdots r_dH(1,1,...,1).
$$
We then conclude

$$
\Lmap_1(r) \leq
r_1r_2r_3\cdots r_d\int_0^{1}\cdots\int_0^1%
\frac{\partial \psi}{\partial u_1}(u_1,u_2,...,u_d)du_1du_2du_3...du_d.
$$
The integral $\int_0^{1}\cdots\int_0^1\frac{\partial \psi}{\partial u_1}(u_1,u_2,...,u_d)du_1du_2du_3...du_d
$ is $1$, as the expected value of the normalized $X_1$, and so we obtain the desired result.
\end{proof}


\begin{proof}[Proof of proposition~\ref{prop:eval}]
    
We need to show that a social evaluation functional~$S$ of the form~(\ref{eq:rank}) satisfies
\begin{eqnarray*}
   X\succcurlyeq_{\mathcal L}X^\prime \Rightarrow S(X)\geq S(X^\prime)
\end{eqnarray*} 
if and only if~$\phi$ is of the form~(\ref{eq:weights}). To show this, we note that 
\begin{eqnarray*}
    S(X)&=&\int_{[0,1]^d}\phi_m(u)^\top \nabla\psi_{\tilde X}(u) \, du\\
     &=&\sum_{j=1}^d \int_{[0,1]^d}\phi_{m,j}(u) \nabla\psi_{\tilde X,j}(u) \, du\\
     &=&\sum_{j=1}^d \int_{[0,1]^d}\left[ \int_{[0,1]^d}\mathds 1\{u\leq r\}\, dm_j(r)\right]\nabla\psi_{\tilde X,j}(u) \, du\\
     &=&\sum_{j=1}^d \int_{[0,1]^d}\left[ \int_{[0,1]^d}\mathds 1\{u\leq r\}\nabla\psi_{\tilde X,j}(u) \, du\right]\, dm_j(r)\\
     &=&\sum_{j=1}^d \int_{[0,1]^d}\mathcal{L}_j(r)\, dm_j(r).
\end{eqnarray*} 
So sufficiency holds for all non-negative measures. Necessity follows by taking one of the measures as a degenerate measure and the other measures as trivial measures. 
\end{proof}


\begin{proof}[Proof of proposition~\ref{prop:trans}]
    It suffices to see that~$\mathcal L_{X}(r)\geq\mathcal L_{X^\prime}(r)$ if and only if
    \begin{eqnarray*}
        \int_{[0,1]^d}\;\mathds 1\{u\leq r\}\left(\nabla\psi_{\tilde X}(u)-\nabla\psi_{\tilde X^\prime}(u)\right)\;du\geq0.
    \end{eqnarray*}
\end{proof}


\begin{proof}[Proof of proposition~\ref{prop:MRT}]
Suppose~$X^\prime$ is obtained from~$X$ through an MRT, so that there is a supermodular and component-wise convex function~$\psi$, such that $\nabla\psi_{\tilde X^\prime}(u)=\nabla\psi_{\tilde X}(u)+\nabla\psi(u)$ holds for all~$u\in[0,1]^d$.
We want to show ~$\mathcal L_{X'} \leq \mathcal L_X$, i.e., $\int\!\!\!\int^r\nabla\psi(u)du \leq 0$.
Consider the first component:
\begin{eqnarray*}
\int_0^{r_1}\int_0^{r_2}...\int_0^{r_d}\frac{\partial \psi(u)}{\partial u_1}du & = & \int_0^{r_2}...\int_0^{r_d}[\psi(r_1,u_2,...,u_d)-\psi(0,u_2,...,u_d)]du_2...du_d \\
& \leq & r_1 \int_0^{r_2}...\int_0^{r_d}[\psi(1,u_2,...,u_d)-\psi(0,u_2,...,u_d)]du_2...du_d\\
&=& r_1\int_0^{1}\int_0^{r_2}...\int_0^{r_d}\frac{\partial \psi(u)}{\partial u_1}du_1u_2....du_d 
\end{eqnarray*}
by convexity.

Now, note that supermodularity implies the function $\int_0^{r_2}...\int_0^{r_d}\frac{\partial \psi(u)}{\partial u_1}du_2....du_d $ is convex with respect to $r_2$, and it is clearly $0$ when $r_2=0$.  Therefore,
$$
\int_0^{r_2}...\int_0^{r_d}\frac{\partial \psi(u)}{\partial u_1}du_2....du_d \leq r_2\int_0^{1}\int_0^{r_3}...\int_0^{r_d}\frac{\partial \psi(u)}{\partial u_1}du_2....du_d.
$$
Similarly, 
$$
\int_0^{r_3}...\int_0^{r_d}\frac{\partial \psi(u)}{\partial u_1}du_3....du_d \leq r_3 \int_0^{1}\int_0^{r_4}...\int_0^{r_d}\frac{\partial \psi(u)}{\partial u_1}du_3....du_d.
$$
Continuing iteratively in this way, we eventually obtain:
\begin{eqnarray*}
\int_0^{r_1}\int_0^{r_2}...\int_0^{r_d}\frac{\partial \psi(u)}{\partial u_1}du &\leq& r_1r_2...r_d\int_0^{1}\int_0^{1}...\int_0^{1}\frac{\partial \psi(u)}{\partial u_1}du\\
&=&r_1r_2...r_d[\mathcal (L_{X'})_1(1,1,...1) - \mathcal (L_X)_1(1,1,...1)]=0  
\end{eqnarray*}

Similar reasoning applies to the other components, and the result follows.
\end{proof}


\begin{proof}[Proof of proposition~\ref{prop:curves}]
(1) Let~$l_X(z_1,z_2)=\alpha$. Since~$l_X$ is a cdf, hence non decreasing in both arguments, then~$z_2:=l_X^{-1}(z_1;\alpha)=\inf\{\zeta; \alpha\leq l_X(z_1,\zeta)\}$ is non increasing in~$z_1$ and non decreasing in~$\alpha$. (2) See Claim~1 in \cite{BT:1966}.
\end{proof}


%

\begin{proposition}(Rawlesian limits of social welfare functions)
    \label{prop:S-Gini}
For a random variable $X$ in $\mathbb{R}$, define, as in Section \ref{subsect: S-Gini}, the social welfare function:
$$
S_\delta(X):=\delta(\delta-1) \int_{[0,1]}(1-r)^{\delta-2}\mathcal{L}_X(r)\, dr.
$$
Then $\lim_{\delta \rightarrow \infty}S_\delta(X) =Q_X(0+) :=\lim_{u \rightarrow 0^+}Q_X(u)$.
\end{proposition}

\begin{proof}[Proof of proposition~\ref{prop:S-Gini}]
  Integrating by parts implies 
  $$
  S_\delta(X)=\delta\int_0^1(1-r)^{\delta-1} Q_X(r)dr=\delta\sum_{i=1}^N\int_{(i-1)/N}^{i/N}(1-r)^{\delta-1} Q_X(r)dr
  $$
  for any positive integer $N$.  As the quantile $Q_X(\cdot)$ is increasing, we then have
  \begin{eqnarray*}
      S_\delta(X) &\leq& \delta\sum_{i=1}^NQ_X(i/N)\int_{(i-1)/N}^{i/N}(1-r)^{\delta-1} dr\\
      &=&\sum_{i=1}^NQ_X(i/N)[-(1-r)^\delta]|^{i/N}_{(i-1)/N}\\
      &=&Q_X(1/N)(1-(1-1/N)^\delta) +\sum_{i=2}^NQ_X(i/N)[(1-\frac{i-1}{N})^\delta-(1-\frac{i}{N})^\delta ]\\
      &=&Q_X(1/N) +M(\delta),
  \end{eqnarray*}
  where the function $M(\delta):=-Q_X(1/N)(1-1/N)^\delta +\sum_{i=2}^NQ_X(i/N)[(1-\frac{i-1}{N})^\delta-(1-\frac{i}{N})^\delta ]$ tends to $0$ as $\delta \rightarrow \infty$.  

  Similarly, we have
      \begin{eqnarray*}
      S_\delta(X) &\geq& \delta\sum_{i=1}^NQ_X((i-1)/N)\int_{(i-1)/N}^{i/N}(1-r)^{\delta-1} dr\\
      &=&Q_X(0+) +m(\delta)
      \end{eqnarray*}
      where $m(\delta) :=-Q_X(0+)(1-1/N)^\delta +\sum_{i=2}^NQ_X((i-1)/N)[(1-\frac{i-1}{N})^\delta-(1-\frac{i}{N})^\delta ]$ tends to $0$ as $\delta \rightarrow \infty$.  
      We therefore have
      $$
      Q_X(0+) \leq \liminf_{\delta \rightarrow \infty}S_\delta(X) \leq \limsup_{\delta \rightarrow \infty}S_\delta(X) \leq Q_X(1/N).
      $$
      As this holds for every integer $N$, and $\lim_{N \rightarrow \infty} Q_X(1/N)= Q_X(0+)$, the result follows.
\end{proof}
The following is a multivariate extension:
\begin{proposition}\label{prop: multivariate s-gini}
Let $X$ be a $d$ dimensional random variable and $Q_X$ its multivariate quantile.  Consider the $j$th term in the sum defining the multivariate S-Gini in Section \ref{subsect: S-Gini}:

$$
S^j_\delta(X):=\delta_j(\delta_j-1) \int_{[0,1]^d}(1-r_j)^{\delta_j-2}\mathcal{L}_j(r)\, dr.
$$
Then
$$
\lim_{\delta_j \rightarrow \infty} S_\delta^j(X) = \int_{[0,1]
^{d-1} }\Pi_{i=1, i\neq j}^d(1-u_i)Q_j(u_1,u_2,...u_{j-1},0+,u_{j+1},...u_d)du_1du_2...du_{j-1}du_{j+1}...du_d,
$$
where $Q_j(u_1,u_2,...u_{j-1},0+,u_{j+1},...u_d) = \lim_{u_j \rightarrow 0^+} Q_j(u)$.
\end{proposition}
Note that the expression is a measure of inequality among agents with lowest rank in the $j$th component, with weighting functions in the other rank variables as in the standard Gini index.


\begin{proof}
    The proof is similar to the argument above.  Without loss of generality, we assume $j=1$. 
 Integrating by parts, we have, for any positive integer $N$:
    \begin{eqnarray*}
S^1_\delta(X)&=&\delta_1\int_{[0,1]^d}(1-r_1)^{\delta_{1}-1}\Big[\int_0^{r_2}\cdots \int_0^{r_d} Q_1(r_1,u_2,...,u_d)du_2...du_d\Big]dr\\
&=& \delta_1\int_{[0,1]^{d-1}}\Big( \sum_{i=1}^N\int_{(i-1)/N}^{i/N}(1-r_1)^{\delta_{1}-1}\Big[\int_0^{r_2}\cdots \int_0^{r_d} Q_1(r_1,u_2,...,u_d)du_2...du_d\Big]dr_1\Big)dr_2...dr_d\\
    \end{eqnarray*}
    Using monotonicity of $Q$ in the first argument, we then have:
   \begin{eqnarray}\label{eqn: upper bound on S}
     &\delta_1\sum_{i=1}^N\int_{(i-1)/N}^{i/N}(1-r_1)^{\delta_{1}-1}\Big[\int_0^{r_2}\cdots \int_0^{r_d} Q_1(r_1,u_2,...,u_d)du_2...du_d\Big]dr_1\nonumber\\
     &\leq \delta_1\sum_{i=1}^N\Big[\int_0^{r_2}\cdots \int_0^{r_d} Q_1(i/N,u_2,...,u_d)du_2...du_d\Big]\int_{(i-1)/N}^{i/N}(1-r_1)^{\delta_{1}-1}dr_1\nonumber\\
&=\sum_{i=1}^N\Big[\int_0^{r_2}\cdots \int_0^{r_d} Q_1(i/N,u_2,...,u_d)du_2...du_d\Big]\Big[(1-\frac{i-1}{N})^{\delta_1} - (1-\frac{i}{N})^{\delta_1}\Big]\nonumber\\
&=\int_0^{r_2}\cdots \int_0^{r_d} Q_1(1/N,u_2,...,u_d)du_2...du_d +M(\delta_1, r_2,...,r_d)
    \end{eqnarray}
where $M(\delta_1,r_2,...,r_d)$ converges monotonically to $0$ as $\delta_1 \rightarrow \infty$.

Therefore, using the monotone convergence theorem, we have 
\begin{eqnarray*}
\limsup_{\delta_1 \rightarrow \infty}S_\delta^1(X) &\leq& \int_{[0,1]^{d-1}} \Big(\int_0^{r_2}\cdots \int_0^{r_d} Q_1(1/N,u_2,...,u_d)du_2...du_d \Big)dr_2...dr_d\\
&=&\int_{[0,1]^{d-1}} \Pi_{i=2}^d(1-u_j)Q_1(1/N,u_2,...,u_d)du_2...du_d 
\end{eqnarray*}
As this holds for all $N$, and $Q_1$ is monotone in the first argument, we can take the limit as $N \rightarrow \infty$ and apply the monotone convergence theorem again to obtain
$$
\limsup_{\delta_1 \rightarrow \infty}S_\delta^1(X) \leq \int_{[0,1]^{d-1}} \Pi_{i=2}^d(1-u_j)Q_1(0+,u_2,...,u_d)du_2...du_d 
$$
A very similar argument to the one used to derive \eqref{eqn: upper bound on S} then implies:

\begin{eqnarray*}
&\delta_1\sum_{i=1}^N\int_{(i-1)/N}^{i/N}(1-r_1)^{\delta_{1}-1}\Big[\int_0^{r_2}\cdots \int_0^{r_d} Q_1(r_1,u_2,...,u_d)du_2...du_d\Big]dr_1\nonumber\\
&\geq \int_0^{r_2}\cdots \int_0^{r_d} Q_1(0+,u_2,...,u_d)du_2...du_d +m(\delta_1, r_2,...,r_d)
\end{eqnarray*}
where $m(\delta)$ converges monotonically to $0$ as $\delta_1 \rightarrow \infty$.  Proceeding as above gives
$$
\liminf_{\delta_1 \rightarrow \infty}S_\delta^1(X) \geq \int_{[0,1]^{d-1}} \Pi_{i=2}^d(1-u_j)Q_1(0+,u_2,...,u_d)du_2...du_d 
$$
which implies the desired conclusion.
\end{proof}

}


\bibliographystyle{abbrvnat}
\bibliography{Lorenz}

\end{document}